\documentclass[preprint]{elsarticle}

\usepackage{hyperref}
\usepackage[normalem]{ulem}

\usepackage{subfigure}

\usepackage{graphicx}
\usepackage{amssymb,amsfonts}
\usepackage{amssymb,amsfonts,amsmath,latexsym,dsfont}
\usepackage[ruled]{algorithm2e}
\usepackage{amsthm}
\usepackage{mathtools}
\usepackage{wasysym}
\usepackage{wrapfig}
\usepackage{float}
\usepackage{chngcntr}
\usepackage{xcolor}
\usepackage{caption}

\usepackage{tikz,pgfplots}
\usetikzlibrary{plotmarks}

\usepackage{fullpage}
\usepackage{graphicx}
\usepackage{mathrsfs}
\usepackage{framed,multirow}
\usepackage{setspace}

\usepackage{booktabs}

\usepackage{soul}
\usepackage{lscape}
\usepackage{caption}

\newdimen\iwidth
\newdimen\iheight

\newcommand{\mc}[3]{\multicolumn{#1}{#2}{#3}}

\newtheorem{remark}{Remark}[section]

\graphicspath{{Fig/}{./}}

\definecolor{darkblue}{rgb}{0.08, 0.25, 0.68}

\renewcommand{\div}{\nabla\cdot}
\newcommand{\grad}{\nabla}

\newcommand{\bfC}{{\bf C}}

\newcommand{\bfF}{{\bf F}}

\newcommand{\bfU}{{\bf U}}

\newcommand{\bfb}{{\bf b}}

\newcommand{\bff}{{\bf f}}

\newcommand{\bfx}{ {\bf x}}

\newcommand{\bfv}{ {\bf v}}

\newcommand{\bfu}{{\bf u}}

\newcommand{\bra}{\left\langle}
\newcommand{\ket}{\right\rangle}

\newtheorem{theorem}{Theorem}[section]

\journal{Journal of \LaTeX\ Templates}

\begin{document}

\begin{frontmatter}

\title{A SIMPLE-Based Preconditioned Solver for the Direct-Forcing Immersed Boundary Method }

\author[mymainaddress]{Rachel Yovel}\ead{yovelr@bgu.ac.il}
\cortext[mycorrespondingauthor]{Corresponding author}
\author[mymainaddress]{Eran Treister}\ead{erant@cs.bgu.ac.il}
\author[yuriaddress]{Yuri Feldman\corref{mycorrespondingauthor}}\ead{yurifeld@bgu.ac.il}

\address[mymainaddress]{Department of Computer Sciences, Ben-Gurion University of the Negev, Beer-Sheva, Israel.}
\address[yuriaddress]{Department of Mechanical Engineering, Ben-Gurion University of the Negev, Beer-Sheva, Israel.}

\begin{abstract}
We present a robust and scalable solver for direct-forcing immersed boundary simulations, based on a preconditioned SIMPLE algorithm. 
The method applies block elimination to the pressure-force coupled system, and utilizes the discrete Laplacian operator as an efficient preconditioner for the resulting Schur complement.
We rigorously demonstrate the spectral equivalence between the Schur complement and the discrete Laplacian, ensuring convergence behavior that is independent of grid resolution and physical parameters. 
This enables accurate, stable, and efficient two-way coupled fluid-structure interaction (FSI) simulations with moving boundaries and significant added-mass effects. These simulations are all executable on standard computing platforms. 
Extensive validation and verification --- including simulations of oscillating, sedimenting, and buoyant spheres, as well as configurations involving multiple immersed bodies --- confirm the solver’s accuracy and efficiency across a broad range of FSI scenarios. 
The proposed approach introduces a novel and accessible framework for immersed boundary simulations requiring strong pressure–force coupling.
\end{abstract}

\begin{keyword}
Moving boundary simulations, Two way coupling FSI, Implicit immersed boundary method, Schur complement, Preconditioning
%\footnote{The first author thanks the Lynn and William Frankel Center for Computer Science at BGU for travel funding.}
\end{keyword}

\end{frontmatter}

%\linenumbers

\section{Introduction}\label{sec:intro}
Immersed boundary methods (IBMs) provide a versatile computational framework for fluid-structure interaction (FSI) problems involving complex or moving geometries, without requiring mesh regeneration or body-fitted grids.
The coupled FSI problem is governed by fluid and structural equations, linked via interface constraints. These interface conditions  enforce displacement and velocity continuity and stress balance, constituting a system of equations coupling velocity, pressure and Lagrangian forces.
The numerical treatment of this system depends on the chosen solution strategy, broadly categorized as monolithic or partitioned. 
Monolithic methods solve all unknowns simultaneously, resulting in intrinsic coupling. In contrast, partitioned methods treat the fluid and structural problems within separate solver frameworks, exchanging information across the interface and enforcing coupling through explicit, semi-implicit, or fully implicit schemes.
Explicit schemes apply fluid and structural updates sequentially within each time step.
Semi-implicit methods introduce limited, typically one-way temporal coupling,  updating one subsystem  based on temporally lagged information from the other.
Fully implicit schemes employ two-way coupling and iterate both subsystems to convergence at each time step.
Explicit approaches can work well in steady or weakly coupled situations. However, they often fail when rapid transients occur, when the fluid and structure has comparable characteristic time scales, or when strong added-mass effects are present.
In these cases, quasi-static assumptions no longer hold. 
To maintain stability and accurately conserve linear and angular momentum, one must use either a monolithic formulation or a partitioned scheme with semi- or fully-implicit coupling.
Examples of such  configurations include  nearly neutrally buoyant particles \cite{tschisgale2018implicit}, stiff elastic membranes \cite{mori2008implicit},  deformable capsules or vesicles modeled as nonlinear membranes \cite{le2009implicit},
 fluid--particle interaction during binder jet 3D printing \cite{wagner2024coupled}, and fluid-structure piezoelectric interaction \cite{li2023numerical}.
For a detailed overview of existing coupling schemes and their applications in incompressible FSI, see the comprehensive review  \cite{fernandez2011coupling}.

Monolithic formulations offer strong coupling and numerical robustness, 
ensuring consistent enforcement of all constraints. 
These include algebraic preconditioning of the Newton iteration applied to large displacement FSI configurations \cite{muddle2012efficient}, direct inversion of an extended Helmholtz operator via LU decomposition \cite{feldman2016extension}, Schur complement reduction with FFT-based Laplacian inversion under periodic boundary conditions \cite{stein2016immersed, stein2017immersed}, Schur complement-based solution of the Stokes system accelerated by geometric multigrid preconditioning \cite{kallemov2016immersed}, and stream-function-vorticity formulations with efficiently precomputed fluid-structure coupling operators \cite{nair2022strongly}.
Despite their robustness, monolithic methods are often computationally expensive due to the need to solve large, indefinite, and ill-conditioned saddle-point systems.
An additional drawback is that the numerical methodology must be developed essentially from scratch for each specific problem, which prevents reuse of legacy codes previously developed for fluid and structural solvers.

These limitations have contributed to the growing popularity of partitioned formulations over recent decades, which offer a practical alternative with greater modularity and scalability across diverse FSI applications. 
Partitioned formulations of IBMs, especially those based on the direct-forcing approach \cite{mohd1997simulations, fadlun2000combined}, have shown particular effectiveness in incompressible flow simulations when employing semi-implicit or fully implicit coupling. 
These methods avoid assumptions about structural dynamics and instead enforce interfacial constraints directly through the momentum equations. 
The constraints are imposed by introducing Lagrangian force terms into the momentum equations to match the desired boundary velocities. 
Pressure and interface force fields act as distributed Lagrange multipliers (DLM), intrinsically coupled to each other and to the velocity field through the governing equations. 
Several works are worth mentioning in this context:
in \cite{wu2009implicit}, it is proposed to combine all unknown Lagrangian forces into a single, fully coupled system to better reflect the parabolic nature of the Navier–Stokes equations.
The works \cite{vanella2009moving, posa2017adaptive}
employ moving-least-squares reconstruction and isoparametric mapping to improve overall accuracy and robustness.
The works \cite{yildiran2024pressure, farah2024improved} aim to achieve high-order accuracy in enforcing the no-slip constraint, and in \cite{koponen2025direct},  Dirichlet and Robin boundary conditions are implemented with high accuracy in sharp corner geometries.

The immersed boundary projection method (IBPM), originally introduced in \cite{taira2007immersed}, has emerged as a particularly influential development within the family of partitioned direct-forcing immersed boundary formulations.
It couples pressure and Lagrangian forces --- both serving as Lagrange multipliers --- within a unified projection framework and supports both semi-implicit and fully implicit \cite{le2008implicit} implementations.
Over the past decade, IBPM has been effectively applied to simulate FSI involving rigid bodies, including sedimentation and particulate flows in two and three dimensions \cite{wang2015strongly, lacis2016stable, ong2022immersed}, as well as biologically inspired problems such as insect flight \cite{li2016efficient}, vesicle dynamics over a wide range of Reynolds numbers \cite{ong2020immersed, ong2021immersed}, and thermal convection involving heat transfer \cite{xu2023efficient}. More recently, integration with SIMPLE \cite{goncharuk2023immersed} and PIMPLE \cite{constant2017immersed, askarishahi2023immersed, farah2024improved} algorithms has considerably extended the method’s applicability, enabling robust simulations of stationary and arbitrarily moving geometries within widely used CFD platforms such as OpenFOAM.

While algorithmic advancements have broadened the applicability of immersed boundary methods, rigorous theoretical analysis remains limited, particularly for direct-forcing approaches. 
Most formal proofs of well-posedness and numerical stability have been established within the framework of variational formulations and finite element methods, typically tailored to fictitious domain methods, as demonstrated in the recent studies \cite{boffi2022existence,decastro2025lagrange} and references therein. 
Related efforts for continuous forcing formulation using finite-difference schemes have focused on proving the unconditional stability of backward Euler and Crank–Nicolson-like discretizations of the nonlinear immersed boundary terms \cite{newren2007unconditionally}, and on developing projection-based preconditioners to accelerate semi-implicit methods through Schur complement reduction \cite{zhang2014projection}. 
Rigorous analysis of direct-forcing IBMs, 
particularly for methods based on finite volume or finite difference discretizations that use distributed Lagrange multipliers (DLMs) to enforce no-slip constraints is especially scarce.
To the best of our knowledge, the only such effort is the preconditioned implicit direct-forcing (PIDF) method  \cite{park2016preconditioned}, which improves no-slip enforcement through iterative force computation. 
However, this method relies on empirically tuned parameters, which limits its generality.
 
The current work contributes to bridging this gap by introducing a new preconditioning strategy for a SIMPLE-based direct-forcing IBM. 
The method solves a regularized saddle-point system to simultaneously update pressure and Lagrangian force corrections, ensuring incompressibility and enforcement of the no-slip constraint.
This saddle-point system is first reduced via a primal Schur complement which is in turn preconditioned by a Laplacian operator.
Building on this formulation, we prove rigorously that the Laplacian is spectrally equivalent to the Schur complement, enabling an efficient and robust preconditioning strategy, and demonstrate numerically that the number of Krylov iterations remains low and  independent of grid resolution and problem parameters. 
We demonstrate the applicability of the method to both moving boundary and two-way coupled FSI problems. 
First, the methodology is verified through simulations of flows induced by a transversely oscillating sphere.
Second, the approach is challenged with more complex configurations involving multiple moving bodies, specifically porous spheres modeled as arrays of rigid sub-spheres. The simulations reveal distinct flow features associated with different array porosities, including variations in drag coefficients, phase shifts, and vortex evolution.
Then we demonstrate the applicability of the approach to two-way FSI simulations, such as sedimenting and buoyant spherical particles.
In all of the above mentioned configurations, the method shows good agreement with experimental and numerical reference data while preserving numerical stability and accuracy. 
Finally, we demonstrate that the method achieves close to linear and sub-linear scalability in computational time and memory usage, respectively, while maintaining notably low memory requirements, making it well-suited for realistic FSI simulations on standard workstations.

\section{Theoretical background}

\subsection{ Governing equations for the dynamics of a spherical particle in a Newtonian flow. }
In compliance with the direct-forcing IBM formulation, the incompressible flow around a solid body is governed by the non-dimensional Navier–Stokes (NS) and continuity equations, supplemented by the kinematic no-slip constraint on the body's surface:
\begin{align}
\label{eq:momentum}
& \frac{\partial \bfu}{\partial t} + (\bfu \cdot \grad) \bfu = -\grad p +\frac{1}{Re} \nabla^2 \bfu +\int_S \mathbf{F}(\boldsymbol{\mathcal{L}})\delta(\boldsymbol{\mathcal{L}}-\mathbf{x})dV_s && \forall \mathbf{x} \in\Omega, \\
\label{eq:divfree_continuous}
& \div \bfu = 0 && \forall \mathbf{x} \in\Omega, \\
\label{eq:kinematic_constraint}
& \bfU (\boldsymbol{\mathcal{L}})=\int_\Omega \bfu(\bfx)\delta(\bfx-\boldsymbol{\mathcal{L}})dV = \bfu_s+\boldsymbol{\omega}_s \times \boldsymbol{\widetilde{\mathcal{L}}} &&
\end{align}
where $\mathbf{u}$ and $p$ are the velocity and pressure fields, $\boldsymbol{\mathcal{L}}$ is the series of Lagrangian points determining the surface $S$ of the immersed body, 
and $\mathbf{x}$ stores the locations of the structured staggered Eulerian grid that hosts the discrete values of $\mathbf{u}$ and $p$ (see Fig. \ref{fig:cell}).
$\mathbf{F}$ is the Lagrangian force density  defined within the shell unit volume $dV_s$ surrounding each Lagrangian point, $\bfU$ is a velocity of Lagrangian point, $\delta$ is the Dirac delta function, and $\mathbf{u}_s$ and $\boldsymbol{\omega}_s$ are the translational and rotational velocities of the rigid body, respectively. $\boldsymbol{\widetilde{\mathcal{L}}}$ is the radius vector connecting the centroid of the rigid body with the surface point $\boldsymbol{\mathcal{L}}$, $dV$ is the unit volume of an Eulerian cell, and $\Omega$ is the whole computational domain occupied by the fluid. The normalized Eqs.~\eqref{eq:momentum}--\eqref{eq:kinematic_constraint} are governed by a single non-dimensional parameter: the Reynolds number, $Re = U_0D/\nu$ where $\nu$ is the kinematic viscosity and $D$ and $U_0$ are the characteristic length and velocity, respectively. The time, pressure, and force density are normalized by utilizing the scales $D/U_0$, $\rho_f U_0^2$, and $\rho_f U_0^2 /D$, respectively, where $\rho_f$ is the fluid density.

\begin{figure}
\centering
\includegraphics[width=0.67\textwidth]{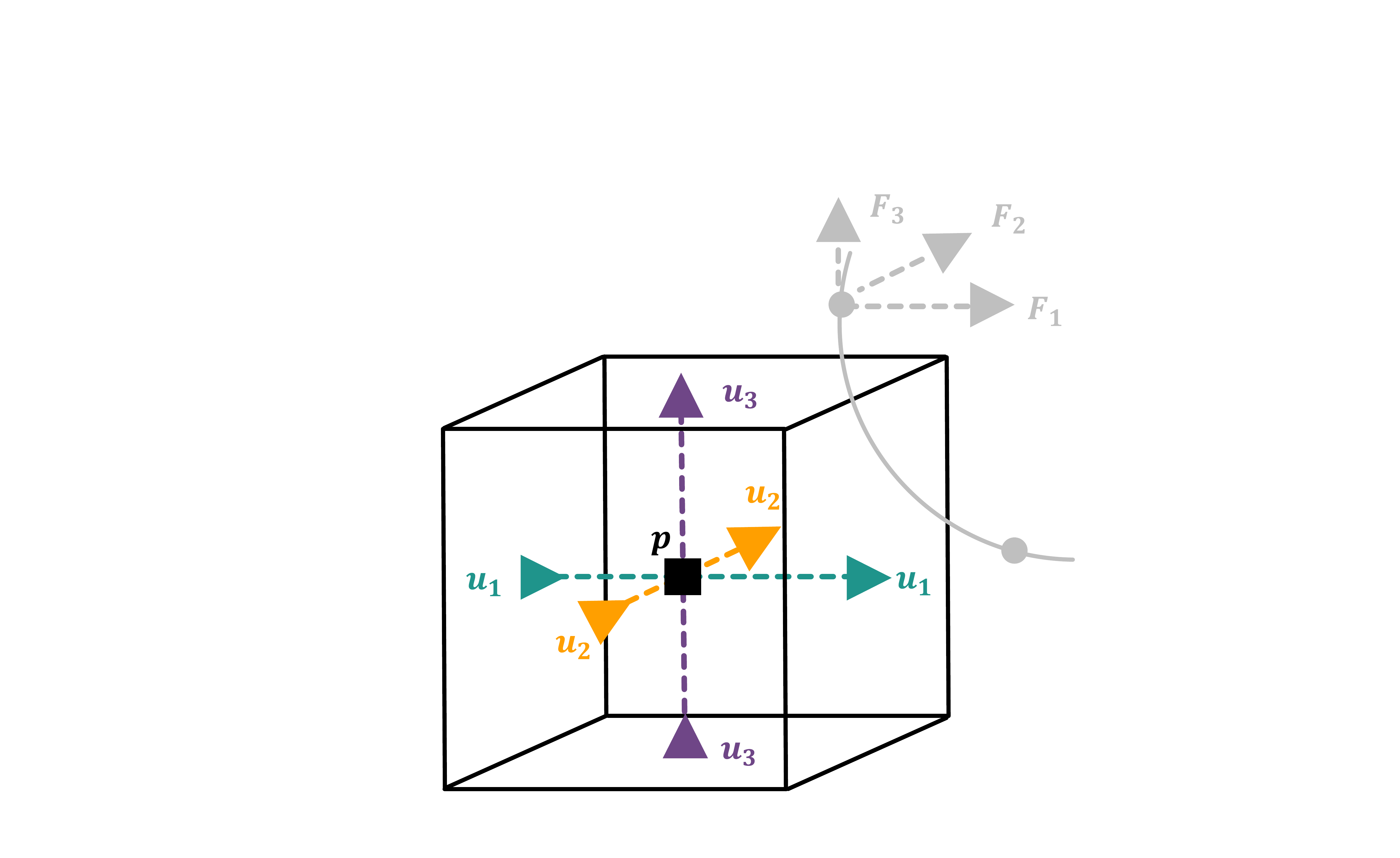}
\caption{An Eulerian discretization cell in 3D and two adjacent Lagrangian points on the surface of the immersed body.}
\label{fig:cell}
\end{figure}

Note that the last term of Eq.~\eqref{eq:momentum} constitutes the dynamic constraint, which defines the Lagrangian force density by which the immersed body acts on the surrounding flow, while Eq.~\eqref{eq:kinematic_constraint} represents the kinematic no-slip constraint at the surface of the immersed body. If the kinematics of the immersed body is prescribed (e.g., the body moves with constant velocity or undergoes transverse oscillations), the configuration corresponds to a one-way coupled FSI problem. In this scenario, the system defined by Eqs.~\eqref{eq:momentum}--\eqref{eq:kinematic_constraint}, complemented by appropriate boundary conditions, is mathematically closed and can be solved throughout the computational domain $\Omega$\footnote{Note that according to the direct-forcing IBM formulation, the flow is also computed inside the immersed body. However, for a rigid body, this internal flow field does not have physical significance.}. In the case of two-way coupled FSI, the system of Eqs.~\eqref{eq:momentum}--\eqref{eq:kinematic_constraint} must be solved simultaneously with Newton’s laws governing the dynamics of the solid body. These non-dimensional equations, obtained by normalizing the moment-of-inertia tensor, time, and force density using $\rho_f D^5$, $D/U_0$, and $\rho_f U_0^2/D$, scales respectively, and formulated within the immersed boundary framework as suggested in \cite{uhlmann2005immersed}, are written for a spherical particle as:
\begin{align}
\label{eq:linear_momentum_sb_norm}
& \dot{\mathbf{u}}_s = \frac{6}{\pi\theta} \left[-\int_S \mathbf{F}(\boldsymbol{\mathcal{L}})\, dV_s + \frac{d}{dt} \int_{V_p} \mathbf{u}\, dV \right] - \frac{\theta - 1}{\theta} Fr^{-2} \mathbf{e}_z, \\[6pt]
\label{eq:angular_momentum_sb_norm}
& \dot{\boldsymbol{\omega}}_s = \frac{60}{\pi\theta} \left[-\int_S \boldsymbol{\widetilde{\mathcal{L}}} \times \mathbf{F}(\boldsymbol{\mathcal{L}})\, dV_s + \frac{d}{dt} \int_{V_p} \mathbf{x} \times \mathbf{u}\, dV \right].
\end{align}
Here, the parameter $\theta = \rho_p / \rho_f$ is the ratio of particle to fluid density, and $Fr = gD / U_0^2$ is the Froude number.
Importantly, for a spherical particle, the moment-of-inertia tensor reduces to a constant scalar, which allows Eqs.~\eqref{eq:linear_momentum_sb_norm}--\eqref{eq:angular_momentum_sb_norm} to be conveniently formulated and solved in the inertial reference frame. However, for a general (non-spherical) particles, it is common to formulate the conservation of angular momentum in a body-attached reference frame. In this case, the coordinate system is typically aligned with the principal axes of the body, for which the moment-of-inertia tensor becomes a constant diagonal matrix. The angular velocity vector $\boldsymbol{\omega}_{cs}$ is then expressed in this body-fixed coordinate system. The transformation between the body-fixed and inertial frames is performed by utilizing rotation matrix which should be updated every time step as detailed in the next section.

\subsection{ Updating the rotation matrix. }
The rotation matrix $\mathbf{R}^{(n+1)}$ is updated using the Rodriguez rotation formula based on the angular velocity vector. At each time step $n$, the angular velocity vector expressed in the body-fixed (co-rotating) reference frame, $\boldsymbol{\omega}_{cs}^{(n)}$, is first normalized to obtain the direction of the rotation axis:
\begin{align}
\label{eq:rotat_axis}
\hat{\boldsymbol{\omega}}_{cs}^{(n)} = \frac{\boldsymbol{\omega}_{cs}^{(n)}}{\|\boldsymbol{\omega}_{cs}^{(n)}\|}.
\end{align}
The rotation angle $\gamma$ over the time interval $\Delta t$ is computed using a second-order Adams--Bashforth scheme:
\begin{align}
\label{eq:theta}
\gamma = \Delta t \left( \frac{3}{2} \|\boldsymbol{\omega}_{cs}^{(n)}\| - \frac{1}{2} \|\boldsymbol{\omega}_{cs}^{(n-1)}\| \right).
\end{align}
A skew-symmetric matrix $\mathbf{W}$ is then constructed from the unit rotation axis $\hat{\boldsymbol{\omega}}_{cs}^{(n)}$:
\begin{align}
\label{eq:Rotat_Matrix}
\mathbf{W} =
\begin{bmatrix}
0 & -\hat{\omega}_{csz}^{(n)} & \hat{\omega}_{csy}^{(n)} \\
\hat{\omega}_{csz}^{(n)} & 0 & -\hat{\omega}_{csx}^{(n)} \\
-\hat{\omega}_{csy}^{(n)} & \hat{\omega}_{csx}^{(n)} & 0
\end{bmatrix},
\end{align}
and the incremental rotation matrix is computed using the Rodrigues formula:
\begin{align}
\label{eq:incement_matr}
\mathbf{R}_\Delta = I_m + \sin\gamma\, \mathbf{W} + (1 - \cos\gamma)\, \mathbf{W}^2,
\end{align}
where $I_m$ is the identity matrix. Finally, the updated rotation matrix is obtained by composing the previous orientation $\mathbf{R}^{(n)}$ with the incremental rotation:
\begin{align}
\label{eq:next_time_step_matr}
\mathbf{R}^{(n+1)} = \mathbf{R}^{(n)} \mathbf{R}_\Delta.
\end{align}
The obtained rotation matrix \( \mathbf{R}^{(n+1)} \) is orthonormal. To express any vector given in the body-attached reference frame in the inertial frame, one multiplies this matrix by the vector. Conversely, multiplying the transpose of \( \mathbf{R}^{(n+1)} \) by any vector expressed in the inertial reference frame transforms it into the coordinates of the body-attached frame.
In particular, our current strategy --- which is also applicable to non-spherical particles --- is to compute all flow variables, Lagrangian force densities, and the translational velocity of the immersed body in the inertial frame, while computing the angular velocity of the immersed body in the body-attached reference frame. The latter is achieved by first transforming the right-hand side of Eq.~\eqref{eq:angular_momentum_sb_norm} into the body-attached reference frame and solving it for \( \boldsymbol{\omega}_{cs}^{(n)} \). Once \( \boldsymbol{\omega}_{cs}^{(n)} \) is computed, it is first used to calculate the rotation matrix \( \mathbf{R}^{(n+1)} \) via Eqs.~\eqref{eq:rotat_axis}--\eqref{eq:incement_matr}, and then transformed back to the inertial frame to evaluate the angular velocity \( \boldsymbol{\omega}_{s} \). This angular velocity is used to compute the velocity of each Lagrangian point \( \mathbf{u}(\boldsymbol{\mathcal{L}}) \) using the kinematic constraint defined in Eq.~\eqref{eq:kinematic_constraint}. The position of the body is updated by adding its linear translation to the angular displacement, which is obtained by multiplying \( \mathbf{R}^{(n+1)} \) with the initial position vector \( \mathbf{x}^{(0)} \). More details regarding the discretization procedure and algorithmic sequence are provided in the following sections.

\subsection{ Interpolation and regularization operators. }
Since the Eulerian and Lagrangian grids do not necessarily coincide, two adjoint operators are need to be introduced to  facilitate the data transfer between them: the regularization operator $R$, which smears the Lagrangian forces onto the underlying  Eulerian grid, and the interpolation operator $R^T$, which interpolates the Eulerian velocities to Lagrangian points. These operators are defined by utilizing the following discrete Dirac delta function, proposed in  \cite{roma1999adaptive}:
\begin{equation}
\delta(r) =
\begin{cases}
\frac{1}{6\Delta r} \left(5-3\frac{|r|}{\Delta r} - \sqrt{-3\left(1-\frac{|r|}{\Delta r}\right)^2 + 1} \right) & \text{ for } 0.5\Delta r < |r| \leq 1.5 \Delta r, \\[1.5em]
\frac{1}{3\Delta r} \left(1+\sqrt{-3\left(\frac{|r|}{\Delta r}\right)^2 + 1} \right) & \text{ for } |r|\leq 0.5\Delta r, \\[1.5em]
0 & \text{ otherwise, }
\end{cases}
\label{eq:delta}
\end{equation}
where $r$ represents the distance from the smeared point and $\Delta r$ represents the cell width in the $r$ direction.
This discrete Dirac delta function gained popularity over the years \cite{uhlmann2005immersed, kempe2012improved, taira2007immersed, goncharuk2023immersed},
due to its compact support of at most three discretization cells in each direction.\footnote{Since the discrete Dirac delta function in Eq. \eqref{eq:delta} provides only a first order of spatial accuracy, other discrete delta functions can be also considered to improve the accuracy (with the price of less favorable sparsity), see \cite{stein2017immersed}.} The regularization and interpolation operators are defined using the discrete Dirac delta function from Eq.~\eqref{eq:delta} as follows:
\begin{align}
\label{eq:interpolation}
\int_\Omega \mathbf{u}(\mathbf{x}) \delta(\mathbf{x} - \boldsymbol{\mathcal{L}})\, dV &= (R^T \mathbf{u})_k \approx \Delta x^3 \sum_{i} \mathbf{u}_i\, \delta^3(\mathbf{X}_k - \mathbf{x}_i), \\
\label{eq:regularization}
\int_S \mathbf{F}(\boldsymbol{\mathcal{L}}) \delta(\boldsymbol{\mathcal{L}} - \mathbf{x})\, dV_s &= (R \mathbf{F})_i \approx \Delta x^3 \sum_{k} \mathbf{F}_k\, \delta^3(\mathbf{x}_i - \mathbf{X}_k),
\end{align}
where \( \mathbf{x}_i \) denotes the coordinates of the \( i \)-th velocity location, and \( \mathbf{X}_k \) denotes the \( k \)-th Lagrangian point.
These discrete summation formulas approximate the integrals under the assumption that the immersed surface has sufficiently smooth local curvature and that the Lagrangian points are uniformly distributed, with spacing comparable to that of the underlying uniform Eulerian grid. The notation \( \delta^3 \) refers to the three-dimensional discrete Dirac delta function. As long as the same delta function is used for both interpolation and regularization, the two operators are transposes of each other in matrix representation. For notational convenience, we denote the regularization operator by \( R \) and the interpolation operator by \( R^T \).

\section{Method: time stepping and the pressure--force-density corrections system}\label{sec: Implement}
\subsection{ One-way coupled moving boundary}
We begin with a brief overview of the SIMPLE-IBM approach introduced in our recent work \cite{goncharuk2023immersed} by expressing the governing Eqs.~\eqref{eq:momentum}-- \eqref{eq:kinematic_constraint} in a semi-discrete form:
\begin{align}
& \frac{3\bfu^{*} - 4\bfu^{(n)}+\bfu^{(n-1)}}{2\Delta t} = -\grad p^{(n)}+\frac{1}{Re}\nabla^2 \bfu^{*} + R\bfF^{(n)} -\left(\bfu\cdot\grad\bfu\right)^{(n)},  \label{eq:ustar}
\\[0.5em]
& \frac{3\bfu'}{2\Delta t} = -\grad p' + R\bfF',  \label{eq:utag}
\\[0.5em]
& R^T\bfu^{(n+1)}= \bfU^\Gamma, \label{eq:unplus1}
\\[0.5em]
& \div\bfu^{(n+1)} = 0,
\label{eq:divfree}
\end{align}
where $R$ is the regularization operator defined in Eq.~\eqref{eq:regularization}, which smears the forces from the Lagrangian points to the adjacent Eulerian velocity locations; $R^T$ is its adjoint interpolation operator, defined in Eq.~\eqref{eq:interpolation}, which interpolates the intermediate velocity field onto the corresponding Lagrangian points; and \( \mathbf{U}^\Gamma \) is the prescribed velocity value on the surface of the immersed body. Here, $\mathbf{u}^*$ is an intermediate velocity field, which is subsequently used to compute the corrections to the pressure $p'$, velocity $\mathbf{u}'$, and Lagrangian force-density field $\mathbf{F}'$. The final stage of the algorithm consists of a correction step,
$p^{(n+1)} = p^{(n)} + p'$ and $\mathbf{F}^{(n+1)} = \mathbf{F}^{(n)} + \mathbf{F}'$, followed by a projection step, $\mathbf{u}^{(n+1)} = \mathbf{u}^* + \mathbf{u}'$, which yields a divergence-free velocity field that also satisfies the kinematic no-slip constraint on the surface of the immersed body. As already mentioned, in the case of one-way moving boundary coupling, the body kinematics is prescribed a priori (i.e., the body position $\mathbf{x}^{(n+1)}$ and velocity $\mathbf{U}^\Gamma$ are known from the kinematic law), and the above system is closed and can be solved subject to appropriate boundary conditions. A standard finite volume central-difference discretization was applied to all the spatial terms, while a second-order backward difference was utilized for discretization of all the temporal derivatives.

When proceeding with the solution, the intermediate velocity field $\bfu^*$ is first obtained from the solution of Eq.~\eqref{eq:ustar}. We note that this velocity field does not necessarily satisfy the no-slip and divergence-free constraints. Second, a divergence is applied to both sides of Eq.~\eqref{eq:utag}, similarly to the original SIMPLE method~\cite{patankar1983calculation}, yielding the modified pressure Poisson equation:
\begin{equation}\label{eq:modified_poisson}
-\nabla^2 p' + \div (R \bfF') = -\frac{3}{2\Delta t} \div \bfu^*,
\end{equation}
where the right-hand side is obtained by utilizing the divergence-free constraint in Eq.~\eqref{eq:divfree}, which implies $\div \bfu^* = -\div \bfu'$. However, Eq.~\eqref{eq:modified_poisson} for the pressure correction is not closed, as it also includes the force correction field. Thus, using Eq.~\eqref{eq:utag} again, we obtain:
\begin{equation}\label{eq:u_n_plus_1}
\bfu^{(n+1)} = \bfu^* + \frac{2}{3} \Delta t (-\grad p' + R\bfF').
\end{equation}
By applying the interpolation operator $R^T$ to both sides of Eq.~\eqref{eq:u_n_plus_1} and substituting the result into Eq.~\eqref{eq:unplus1}, we obtain the following system of equations governing the corrections for the pressure and force density fields:
\begin{equation}\label{eq:Saddle}
\begin{bmatrix}
L & B^T\\
B & -C
\end{bmatrix}
\begin{bmatrix}
    p' \\
    \bfF'
\end{bmatrix}
=
\begin{bmatrix}
G^T G & G^T R \\
R^T G & -R^T R
\end{bmatrix}
\begin{bmatrix}
    p' \\
    \bfF'
\end{bmatrix}
=
\begin{bmatrix}
    RHS_{p'} \\
    RHS_{\bfF'}
\end{bmatrix}
\end{equation}
where $G$ is the discretized gradient operator, $RHS_{p'} = -\frac{3}{2\Delta t} \div \bfu^*$, and $RHS_{\bfF'} = \frac{3}{2\Delta t}(R^T \bfu^* - \bfU^\Gamma)$. 
Once the system in Eq.~\eqref{eq:Saddle} is solved, we update the velocity, pressure, and Lagrangian force density fields and proceed to the next time step.\footnote{In principle, internal fixed-point iterations may be applied to improve accuracy. However, for one-way coupling configurations and a sufficiently small time step ($\Delta t = 10^{-3}$), the desired accuracy is achieved in a single step~\cite{goncharuk2023immersed}.} 
The system of equations in Eq.~\eqref{eq:Saddle} constitutes a large, regularized saddle-point problem, whose efficient solution is a major bottleneck of the method described for one-way moving boundary, and in the next subsection we derive a similar equation that constitutes the bottleneck for two-way FSI configurations as well. 
The efficient solution of this saddle-point problem is a core contribution of the present study and is addressed in detail in the following sections.

\subsection{Two-way coupled FSI}
The conceptual difference between one-way coupled moving boundary simulations and two-way coupled FSI lies in the fact that the latter requires coupling the NS equations, governing the fluid dynamics, with the structural equations governing the dynamics of the solid body — specifically, Eqs.~\eqref{eq:linear_momentum_sb_norm} and \eqref{eq:angular_momentum_sb_norm}, which represent the conservation of linear and angular momentum of the solid particle. Remarkably, both equations include contributions from fluid inertia, commonly referred to as the added mass effect, which accounts for the rate of change of linear and angular momentum of the fluid enclosed within the volume occupied by the immersed body~\cite{uhlmann2005immersed}. These contributions appear as volume integrals over the domain occupied by the rigid body and must be evaluated accurately to ensure physical fidelity of the simulation. In the current study, we adopt the voxel-based strategy proposed by \cite{kempe2012improved}, wherein the particle surface is approximated using a Cartesian grid. In this approach, each Eulerian cell is assigned a partial volume fraction $\alpha_{i}^{(n)} \in [0,1]$ that represents the fraction of the cell lying within the particle at time $t^{(n)}$. The computation of $\alpha_i^{(n)}$ follows the linear reconstruction scheme introduced in~\cite{kempe2012improved}, which uses the signed distance function to approximate the volume fraction of the particle within each cell. Accordingly, the volume integrals over the particle domain, appearing in the linear and angular momentum balance equations, are approximated as:
\begin{subequations}
\label{eq:volume_integrals}
\begin{align}
&\int_{V_p^{(n)}} \mathbf{u}^{(n)} \, dV \approx \sum_{i \in V_p^{(n)}} \alpha_i \mathbf{u}_i^{(n)} \, \Delta V, \\[0.5em]
&\int_{V_p^{(n)}} \mathbf{x}^{(n)} \times \mathbf{u}^{(n)} \, dV \approx \sum_{i \in V_p^{(n)}} \alpha_i \mathbf{x}_i^{(n)} \times \mathbf{u}_i^{(n)} \, \Delta V.
\end{align}
\end{subequations}

\subsubsection{Iterative Coupling Algorithm}\label{sec:internal_iterations}
As already mentioned, our method belongs to the family of partitioned approaches; that is, the fluid and structural governing equations are solved separately and coupled through internal iterations at each time step to achieve strong fluid-structure interaction. At each time step, we employ an iterative predictor-corrector scheme. The algorithm, inspired by~\cite{wang2015strongly}, is tailored specifically to our immersed boundary framework and proceeds as follows:

\paragraph{Step 1: Explicit prediction of particle position and orientation}
The particle's new position and orientation are explicitly predicted with second-order accuracy in time using a backward differentiation formula:
\begin{equation}
\mathbf{x}^{(n+1)} =
\frac{1}{3}(4\mathbf{x}^{(n)} - \mathbf{x}^{(n-1)}) + \frac{2}{3}\Delta t\, \mathbf{u}_s^{(n)} + \mathbf{R}^{(n+1)} \mathbf{x}^{(0)}.
\label{eq:particle_position}
\end{equation}
These quantities remain fixed throughout the subsequent inner iterations to maintain computational efficiency.

\paragraph{Step 2: Eulerian velocity prediction}
We solve the discretized Navier–Stokes equations once to obtain an initial prediction, $\mathbf{u}^{(k_0)}$, for the Eulerian velocity field:
\begin{equation}
\frac{3\mathbf{u}^{(k_0)} - 4\mathbf{u}^{(n)} + \mathbf{u}^{(n-1)}}{2\Delta t} = -\nabla p^{(n)} + \frac{1}{Re} \nabla^2 \mathbf{u}^{(k_0)} + R\mathbf{F}^{(n)} - (\mathbf{u} \cdot \nabla \mathbf{u})^{(n)}.
\label{eq:ustar_two_way}
\end{equation}
At this stage, incompressibility and no-slip conditions are not yet enforced.

\paragraph{Step 3: Iterative correction procedure}
An iterative correction loop indexed by \(k \geq 0\) is initiated:

\begin{itemize}
\item \textbf{Translational velocity update:}
\begin{align}
&\mathbf{u}_s^{(k)} = \frac{1}{3}(4\mathbf{u}_s^{(n)} - \mathbf{u}_s^{(n-1)})
- \frac{4\Delta t}{\pi\theta} \int_S \mathbf{F}^{(n,k)}(\boldsymbol{\mathcal{L}}^{(n+1)})\, dV_s \nonumber \\[0.5em]
&\quad + \frac{2}{\pi\theta}\left(
3\int_{V_p^{(n+1)}} \mathbf{u}^{(k)}\, dV
- 4 \int_{V_p^{(n)}} \mathbf{u}^{(n)}\, dV
+ \int_{V_p^{(n-1)}} \mathbf{u}^{(n-1)}\, dV
\right) \nonumber \\[0.5em]
&\quad - \frac{2\Delta t}{3} \cdot \frac{\theta - 1}{\theta} Fr^{-2} \mathbf{e}_z.
\label{eq:bdy_linear_vel_iter}
\end{align}
Note that this discrete form of Eq.~\eqref{eq:linear_momentum_sb_norm}, which governs the conservation of linear momentum of the solid body, is obtained here for the spherical particle  by discretizing all temporal derivatives, using a second-order backward finite difference.

\item \textbf{Rotational velocity update:}
\begin{align}
&\boldsymbol{\omega}_s^{(k)} =
\frac{1}{3} \left( 4 \boldsymbol{\omega}_s^{(n)} - \boldsymbol{\omega}_s^{(n-1)} \right)
- \frac{40\Delta t}{\pi\theta} \int_S \widetilde{\boldsymbol{\mathcal{L}}}^{(n+1)} \times \mathbf{F}^{(n,k)}(\boldsymbol{\mathcal{L}})\, dV_s \nonumber \\[0.5em]
&\quad + \frac{20}{\pi\theta} \left(
3\int_{V_p^{(n+1)}} \mathbf{x}^{(n+1)} \times \mathbf{u}^{(k)}\, dV
- 4\int_{V_p^{(n)}} \mathbf{x}^{(n)} \times \mathbf{u}^{(n)}\, dV
+  \int_{V_p^{(n-1)}} \mathbf{x}^{(n-1)} \times \mathbf{u}^{(n-1)}\, dV
\right).
\label{eq:bdy_angular_vel_iter}
\end{align}
Similarly, this discrete form of Eq.~\eqref{eq:angular_momentum_sb_norm}, which governs the conservation of angular momentum of the solid body, is obtained here for the spherical particle  by discretizing all temporal derivatives, using a second-order backward finite difference.

\item \textbf{Lagrangian boundary velocity evaluation:}
This step is required for updating the RHS of the system in Eq.~\eqref{eq:pressure_force_iterative_system}.
\begin{equation}
\mathbf{U}^{(k)}(\boldsymbol{\mathcal{L}}^{(n+1)}) = \mathbf{u}_s^{(k)} + \boldsymbol{\omega}_s^{(k)} \times \widetilde{\boldsymbol{\mathcal{L}}}^{(n+1)}.
\label{eq:kinematic_constraint_iterative}
\end{equation}

\item \textbf{Pressure and force-density correction:}
The coupled pressure and velocity corrections are computed at each iteration. As the iteration converges, the RHS of the system below approaches zero, and so do the correction values, since the operator is not singular.
\begin{equation}
\begin{bmatrix}
G^T G & G^T R \\[0.5em]
R^T G & -R^T R
\end{bmatrix}
\begin{bmatrix}
p' \\[0.5em]
\mathbf{F}'
\end{bmatrix}
=
\begin{bmatrix}
-\frac{3}{2\Delta t} \nabla \cdot \mathbf{u}^{(k)} \\[0.5em]
\frac{3}{2\Delta t} (R^T \mathbf{u}^{(k)} - \mathbf{U}^{(k)})
\end{bmatrix}.
\label{eq:pressure_force_iterative_system}
\end{equation}

\item \textbf{Fields update and convergence check:}
This is essentially a Gauss--Seidel step with $\xi$ acting as an under- or over-relaxation parameter. Our numerical experiments revealed that $\xi$ can range between 0.6 for the smallest  and 1.1 for the largest values of $\rho_p/\rho_f$ ratio.
\begin{equation}
\begin{aligned}
&p^{(n,k+1)} = p^{(n,k)} + \xi p', \\
&\mathbf{F}^{(n,k+1)} = \mathbf{F}^{(n,k)} + \xi \mathbf{F}', \\
&\mathbf{u}^{(k+1)} = \mathbf{u}^{(k)} + \xi\frac{2\Delta t}{3} \left( -\nabla p' + R\mathbf{F}' \right).
\end{aligned}
\end{equation}
Iterations (i.e., $Step~3$) continue until the relative \(L_1\)-norm of the incremental force-density correction falls below a specified threshold (typically \(10^{-3}\)), ensuring accurate enforcement of divergence-free and no-slip conditions:
\begin{equation}
\frac{\|\mathbf{F}'\|_1}{\|\mathbf{F}^{(n,k+1)}\|_1} < 10^{-3}.
\label{eq:convergence_criterion}
\end{equation}
Upon convergence, we proceed to the subsequent time step. 
\end{itemize}

This iterative scheme ensures consistent enforcement of the incompressibility and no-slip constraints while maintaining modularity between the fluid and solid solvers.
Similarly to the one-way moving boundary case, the system \eqref{eq:pressure_force_iterative_system} constitutes the bottleneck of the entire method and our main contribution is the preconditioning approach to address it, presented in the next section.

\section{Method: preconditioning}

To motivate the need of developing an efficient preconditioner, we first explain why the main computational bottleneck in our framework is the repeated solution of the regularized saddle-point linear system coupling the pressure and force-density corrections.
First, we explain the computational framework of our solver.
As described in the previous section, the fluid and structural dynamics equation are solved separately, and then we either advance directly to the next time step (for one-way coupled simulations) or perform fixed-point internal iterations (for two-way coupled simulations). 
\begin{enumerate}
\item \textbf{Fluid dynamics solver:}
The solution of the discretized Navier–Stokes equations, i.e., Eq.~\eqref{eq:ustar} or Eq.~\eqref{eq:ustar_two_way}, does not present significant computational difficulties, as the governing discrete operators exhibit diagonal dominance for sufficiently small time steps (due to the presence of temporal derivative terms). 
In the present study, we have utilized the direct method proposed by \cite{lynch1964direct}, as implemented in \cite{vitoshkin2013direct}, for factorizing the Helmholtz operator. 
However, previous experience suggests that standard Krylov subspace solvers, such as BiCGStab or GMRES, can also be employed effectively, typically converging within two or three iterations. 
Another feasible approach involves approximating the inverse Helmholtz operator via an $N$-th order Taylor series expansion, exploiting its symmetry and positive definiteness for an appropriate choice of the time step $\Delta t$ and approximation order $N$ \cite{perot1993analysis}. 
\item \textbf{Structural dynamics solver:}
in addition to the fluid equations, the system also includes equations governing structural dynamics --- specifically Eqs.~\eqref{eq:particle_position}, \eqref{eq:bdy_linear_vel_iter}, and \eqref{eq:bdy_angular_vel_iter} --- relevant for two-way coupled fluid-structure interaction. 
These equations represent initial-value ordinary differential equations (ODEs), and their numerical solution is straightforward and computationally inexpensive. 
Currently, we discretize them using a second-order backward finite-difference scheme, although alternative methods, such as those from the Runge-Kutta family, or even analytical solutions could equally be considered.
\item \textbf{The pressure and force-density corrections coupled system:}
Ultimately, the computational bottleneck and the cornerstone of our numerical framework is the efficient solution of the regularized saddle-point linear system given by either Eq.~\eqref{eq:Saddle} for one-way coupling or Eq.~\eqref{eq:pressure_force_iterative_system} for two-way coupling scenarios. Since this system appears at each iteration or time step, its efficient solution has a significant impact on the overall computational cost and stability of the algorithm.
\end{enumerate}

Let us examine more closely the challenges associated with solving the saddle-point system. A common approach is to solve the full coupled system directly, as originally proposed in the IBPM by \cite{taira2007immersed}. In their formulation, the authors derived a modified Poisson equation by approximating the inverse Laplacian operator via a Taylor series expansion in time. Notably, when a first-order approximation is applied, it leads to a saddle-point system with the same left-hand-side as in the present work. 
%The generalized saddle-point matrix of Eq. \eqref{eq:Saddle} was presented there as $Q^T Q$, with $Q = \begin{bmatrix}
%G & R
%\end{bmatrix}$,
%hence it is clear that the matrix is symmetric, and at least positive semi-definite. 
The authors note that given a careful choice of time step and Lagrangian points, the modified Poisson operator is symmetric and positive definite, and hence the system can be solved using the conjugate gradient method. 
Nevertheless, the efficiency of the iterative solution is not directly addressed, and there is no guarantee for iteration count independent on grid resolution and physical parameters.
A key challenge of this approach lies in the involvement of the Laplacian operator, which inherently governs the system. 
The Laplacian is known for its poor conditioning due to the elliptic nature of the diffusion process it represents. 
It spans multiple spatial scales, resulting in a wide eigenvalue spectrum and inherently poor conditioning. 
This limitation becomes more pronounced at higher resolutions, as finer grids capture smaller spatial scales, leading to unbounded growth in the condition number. 
Despite the existence of efficient method for the solution of \emph{standard}  Poisson problems,
such as multigrid method and specially tailored direct solvers, these methods
 are inefficient when applied directly to the  \emph{modified} Poisson equation arising in our saddle-point formulation. 
These considerations underscore the need for an efficient preconditioner to enable robust and rapid convergence of solutions to the saddle-point systems described by Eq.~\eqref{eq:Saddle} or Eq.~\eqref{eq:pressure_force_iterative_system}, as well as convergence theory. 
Particularly, if this preconditioner enables reduction of modified to standard Poisson problems, it will allow the efficient use of many existing methods.
The development of such a preconditioner, along with both formal and numerical validation of its effectiveness, is the focus of the following sections.

\subsection{The Laplacian as a preconditioner}
The idea of constructing an efficient preconditioner originates from forming the \textit{primal} Schur complement:
\begin{equation}\label{eq:SchurPrimal}
S_p \coloneqq L + B^T C^{-1} B,
\end{equation}
and subsequently solving the decomposed system of equations, first for the pressure correction $p'$ and then for the force-density correction $\mathbf{F}'$:
\begin{align}\label{eq:SchurPrimalP}
S_p p' & = RHS_{p'} + B^T C^{-1} RHS_{\mathbf{F}'}, \\
\label{eq:SchurPrimalF}
\mathbf{F}' & = C^{-1} (B p' - RHS_{\mathbf{F}'}).
\end{align}
To reduce computational cost, we adopt the scalar approximation \( C \approx \frac{1}{2}I_m \) proposed in \cite{goncharuk2023immersed}
\footnote{This scalar approximation was achieved in \cite{goncharuk2023immersed} by lumping the matrix $R^T R$.}
, thereby avoiding the explicit inversion of \( C \) by using \( C^{-1} \approx 2I_m \). 
This leads to the following approximation of the primal Schur complement:
\begin{equation} \label{eq:tildeSp}
S_p \approx \tilde{S}_p \coloneqq L + 2 B^T B.
\end{equation}
Substituting this approximation into Eqs.~\eqref{eq:SchurPrimalP} and \eqref{eq:SchurPrimalF} yields the simplified system:
\begin{align}\label{eq:ApproxSchurPrimalP}
\tilde{S}_p p' & = RHS_{p'} + 2 B^T  RHS_{\mathbf{F}'}, \\
\label{eq:ApproxSchurPrimalF}
\mathbf{F}' & = 2(B p' - RHS_{\mathbf{F}'}).
\end{align}
However, both \( S_p \) and \( \tilde{S}_p \) remain large and ill-conditioned, primarily due to the system’s dependence on the Laplacian operator. To address this, we explicitly employ the Laplacian \( L \) itself as a preconditioner. In the following sections, we demonstrate that Eq.~\eqref{eq:SchurPrimalP} can be efficiently solved using a Krylov subspace method preconditioned by \( L \), leveraging the spectral equivalence between \( S_p \) and \( L \). Specifically, in Section~\ref{sec:thm}, we show that the eigenvalues of the preconditioned operator \( L^{-1}S_p \) are contained within a narrow, bounded interval that is independent of physical parameters and grid resolution. Provided that \( \tilde{S}_p \) closely approximates \( S_p \), similar spectral properties hold approximately for the operator \( L^{-1}\tilde{S}_p \). This spectral equivalence leads to rapid convergence of Eq.~\eqref{eq:ApproxSchurPrimalP} when solved with the Laplacian preconditioner, as verified numerically in Section~\ref{sec:results}. Since the pressure--force-density correction system constitutes the main computational bottleneck, this preconditioning strategy substantially improves overall computational efficiency and demonstrates favorable scalability.

\begin{remark}
There is abundance of literature on block-preconditioning for saddle point matrices as appears in Eq. \eqref{eq:Saddle} or \eqref{eq:pressure_force_iterative_system}, mainly by Schur complement preconditioning.
Since forming and inverting the exact Schur-complement is computationally expensive, broad research is dedicated to the search of cheap approximations for the inverse of the Schur-complement, see \cite{elman2014finite}.
Our method offers an efficient approximation for the primal Schur complement for a certain class of saddle point matrices, and can be implemented either by the steps described in Section \ref{sec: Implement} or by a corresponding block lower triangular Schur complement preconditioner.
\end{remark}

\begin{remark}
It is worth noting that the condition number of the matrix $R^T R$ can grow significantly when the Lagrangian discretization is substantially over- or under-resolved relative to the Eulerian grid. Although the theoretical invertibility of $R^T R$ is ensured as long as the Lagrangian points are distinct (since, in this case, $R$ has linearly independent columns), practical measures should be taken to avoid near-singularity. A common choice in the literature is to use comparable spacing between the two grids; see, e.g., \cite{taira2007immersed}. Near-singularity may also occur in solid--solid or solid--boundary interactions. In such cases, a sufficiently fine grid is used to keep the points distinct, and the model includes a repulsion force that enforces point separation. An extended discussion of the grid-spacing effects on the $R^T R$ matrix--vector product approximation and the influence of boundary proximity on the preconditioner performance is provided in \ref{A1} and \ref{A2}, respectively.
\end{remark}

\subsection{Efficient matrix-free implementation}
In this subsection, we describe the current matrix-free implementation. 
First, recall that the operator $B^T$ consists of the (negative) divergence applied to the discrete Dirac delta function weights obtained from Eq.~\eqref{eq:delta}. 
A matrix-free implementation allows efficient shared-memory parallelization (implemented here with OpenMP and executed with 48 threads) and enables rapid recomputation of these operators at each time step. 
Second, at each step of the preconditioned Krylov iteration, only a matrix-vector product with the preconditioned operator is required. 
We employ left preconditioning, which is mathematically equivalent to solving Eq.~\eqref{eq:ApproxSchurPrimalP} after multiplying both sides from the left by $L^{-1}$. However, we do not explicitly form the inverse operator $L^{-1}$, nor do we repeatedly apply a Laplacian solver to multiple columns. Instead, noting that the preconditioned system can be represented as:
\begin{equation*}
L^{-1}\tilde{S}_p = I_m + 2L^{-1}B^TB,
\end{equation*}
where $I_m$ is an identity matrix of appropriate dimension, we observe that the product of $L^{-1}\tilde{S}_p$ with a vector $\mathbf{v}$ can be efficiently obtained by applying the direct solver of \cite{lynch1964direct} only once to the vector $B^T(B\mathbf{v})$. Employing Intel MKL routines\footnote{We convert $B$ and $B^T$ into CSR format at each time step to facilitate efficient use of Intel MKL routines.} allows efficient computation of the sparse matrix-vector products. Furthermore, due to the scalar approximation for $C^{-1}$, computing the second term of Eq.~\eqref{eq:ApproxSchurPrimalP} and the entire Eq.~\eqref{eq:ApproxSchurPrimalF} is trivial, as both involve only a single matrix-vector product. 
Importantly, our matrix-free implementation requires only a single application of the operator $L^{-1}$ to the vector $B^T(B\mathbf{v})$ per iteration or time step. 
Consequently, it incurs no additional computational cost compared to non-preconditioned solvers for Eqs. \eqref{eq:Saddle} or \eqref{eq:pressure_force_iterative_system}.

Finally, we discuss the factorization of the preconditioner $L$. 
Lynch~\cite{lynch1964direct} demonstrated that any matrix expressible as a sum of Kronecker products of smaller matrices with identity matrices can be factorized efficiently using eigen-decomposition of each smaller matrix. 
As shown therein, the resulting system can be solved with a computational complexity of $O(n^{4/3})$ for a regular 3D grid. 
The computational cost can be further reduced by applying the Thomas TDMA algorithm along one of the three coordinate directions. 
For transient problems, this approach is highly efficient, as the associated operators can be factorized once in advance and reused at every time step~\cite{vitoshkin2013direct}.

\section{Theory: a spectral bound on the preconditioned system}\label{sec:thm}

In this section, we analyze the use of the Laplacian \( L \) as a preconditioner for the Schur complement \( S_p \) of the coupled pressure--force-density system \eqref{eq:Saddle} or \eqref{eq:pressure_force_iterative_system}.
A good preconditioner should resemble the original matrix, in the sense of spectral equivalence, i.e., the spectrum of the preconditioned system should be bounded.
Theorem \ref{thm:thm} shows this property.

\begin{theorem} \label{thm:thm}
For every generalized saddle-point matrix of the form
\begin{equation}\label{eq:SaddleSym}
\begin{bmatrix}
L & B^T \\
B & -C
\end{bmatrix}
=
\begin{bmatrix}
G^TG & G^T R \\
R^T G & -R^T R
\end{bmatrix},
\end{equation}
where $G$ is an $n\times m$ matrix and $R$ an $n\times k$ matrix\footnote{Typically, $n>m>k,$ but this assumption is not necessary for the Theorem's statement.},
the leading block $L$ is spectrally equivalent to the primal Schur complement $S_p$ from Eq. \eqref{eq:SchurPrimal}, in the sense
\begin{equation} \label{eq:eigsPrecondtioned}
\text{\upshape{spec}}(L^{-1}S_p) \subseteq [1,2].
\end{equation}
\end{theorem}

\begin{proof}
It is readily seen that
\begin{equation*}
L^{-1}S_p = I_m + L^{-1} B^T C^{-1} B
\end{equation*}
where $I_m$ stands for the identity matrix of size $m\times m$.
Hence, it is sufficient to show that
\begin{equation*}
\text{spec}(L^{-1} B^T C^{-1} B)\subseteq [0,1],
\end{equation*}
or equivalently, that the generalized eigenvalue problem
\begin{equation}\label{eq:generalizedShifted}
B^T C^{-1} B \bfv = \lambda L \bfv
\end{equation}
has only solutions that satisfy $\lambda\in[0,1]$.
Observe that
\[
B^T C^{-1} B = G^T R (R^T R)^{-1} R^T G = G^T P G
\]
where $P \coloneqq R (R^T R)^{-1} R^T$ is an orthogonal projector (namely, $P^2=P$ and $P^T=P$).
To bound $\lambda$, we bound the generalized Rayleigh quotient
\begin{equation}\label{eq:Rayleigh01}
\frac{\bra B^TC^{-1}B\bfv, \bfv \ket}{\bra L\bfv,\bfv \ket} = \frac{\bra G^T P G \bfv, \bfv \ket}{\bra G^T G \bfv, \bfv \ket} = \frac{\bra PG \bfv, G \bfv \ket}{\bra G \bfv, G \bfv \ket},
\quad
\bfv\in\mathbb{R}^m.
\end{equation}
Since it is an orthogonal projector, the only eigenvalues of $P$ are 0 and 1, and $\text{Null}(P) = \text{Range}(P) ^ \perp,$ see, e.g., \cite{saad2011numerical}.
Particularly, every vector can be uniquely decomposed into a sum of orthogonal vectors in the nullspace and in the range of $P$, e.g.,
we can write $G \bfv = \bfv_1 + \bfv_2$, where $\bfv_1\in \text{Null}(P)$, $\bfv_2\in \text{Range}(P)$ and $\bra \bfv_1,\bfv_2 \ket = 0.$
Consequently,
\begin{equation}\label{eq:ProjectionNorm}
\bra PG \bfv, G\bfv \ket = \bra P(\bfv_1+\bfv_2), \bfv_1+\bfv_2 \ket = \bra P\bfv_1, \bfv_1 \ket + 2 \bra P\bfv_1, \bfv_2 \ket + \bra P\bfv_2, \bfv_2 \ket = \bra \bfv_2,\bfv_2 \ket = \|\bfv_2\|^2,
\end{equation}
where the third equality holds because $P\bfv_1 = 0$ and $P\bfv_2 = \bfv_2$.
Finally, we substitute \eqref{eq:ProjectionNorm} into \eqref{eq:Rayleigh01}, which yields
\begin{equation}
0 \leq \frac{\bra PG \bfv , G\bfv \ket}{\bra G\bfv, G \bfv \ket} = \frac{\|\bfv_2\|^2}{\| \bfv_1+\bfv_2 \|^2} = \left(\frac{\|P(\bfv_1+\bfv_2)\|}{\| \bfv_1+\bfv_2 \|}\right)^2  \leq \|P\|^2_2 =\left(\rho(P)\right)^2 = 1
\end{equation}
where $\rho(P)$ is the spectral radius of $P$, and since $P$ is symmetric, it is equal to $\|P\|_2$.
\end{proof}

The importance of Theorem \ref{thm:thm} stems from its generality, i.e., for the spectral equivalence to hold, any discretisation of the gradient and any regularization operator can be considered. The expected scalability and robustness of the method is therefore proved for a large variety of cases, with symmetry being the only limitation, i.e., the divergence must be the transpose of the gradient, and the interpolation must be the transpose of the regularization.

The assumption of symmetry is not particularly restrictive, as it is natural to take the same discrete Dirac delta functions for the regularization and interpolation, and the same discretization for the gradient and divergence.
%However, the pressure's Laplacian, having Neumann BC is singular, as for incompressible flow, the pressure field is determined only up to an arbitrary constant.
%Many modifications are possible to prevent this singularity, one of them is adding a Dirichlet point within the computational domain, and in this case the resulting Laplacian might have several entries that breaks the symmetry.
In practice, however, some configurations might lead to a non-symmetric Laplacian. These include, for example, computational domains with all no-slip boundary conditions, which require adding a Dirichlet point within the computational domain to prevent singularity of the pressure-correction Laplacian. Nevertheless, even when the Laplacian slightly differs from $G^T G$, the spectral result of Theorem~\ref{thm:thm} seems to hold in practice, as demonstrated in the next section.

\begin{remark} \label{remark:convergenceRate}
As an immediate corollary of Theorem \ref{thm:thm}, the error iteration matrix $T = I_m-L^{-1}S_p$ of the preconditioned system satisfies $\rho(T)\leq 1$. When this inequality is strong, convergence is guaranteed.
A Krylov iterative method typically further accelerates the convergence and might converge in just a few iterations even when the spectral radius of the error iteration matrix is even greater than 1.
However, the convergence rate of a Krylov iterative method is determined by the scattering of the eigenvalues, which can depend on the specific configuration.
In typical cases, the number of Lagrangian points is very small compared to the number of Eulerian cells.
In such a case,
\begin{equation}
\tilde{S}_p = L + 2B^TB\approx L
\end{equation}
and the error iteration matrix has a large null-space, which implies that the eigenvalues are clustered and the Krylov iteration count is low.
In the next section, we demonstrate that even when the number of Lagrangian points grows --- by taking multiple bodies --- the observed iteration count remains low.
\end{remark}

\begin{remark}
A similar result can be proved for a weighted version of the saddle-point matrix, such as
\begin{equation}
\begin{bmatrix}
G^T W G & G^T W R \\
R^T W G & -R^T W R
\end{bmatrix},
\end{equation}
where $W$ is SPD, by taking $\tilde{G}=W^{\frac{1}{2}}G$ and $\tilde{R}=W^{\frac{1}{2}}R$.
Since such a weighted system appears in \cite{taira2007immersed}, it gives rise to a utilization of a similar preconditioning approach in the context of projection methods.
\end{remark}

\section{Results}\label{sec:results}

In this section we demonstrate the efficiency of our method in solving various problems.
In Subsection \ref{subsec:verification}
we present several model problems in 3D, and verify our method by comparing the simulation results to existing references.
In subsection \ref{subsec:efficiency} we demonstrate the computational efficiency of our preconditioning approach in 2D and in 3D.

\subsection{Verification study}
\label{subsec:verification}

\subsubsection{Model problem 1: a transversely  oscillating solid sphere (one-way coupling)}
\label{subsubsec:OscillatingSphere}

A transversely oscillating solid sphere of diameter $D$ placed within a $4D\times 4D \times 6D$ box filled with a quiescent fluid (see Fig. \ref{fig:ball}) is considered. The sphere oscillates vertically with velocity and position given by:
\begin{equation}
U_z = U_{max}\sin({\omega T}), \qquad Z=-\frac{U_{max}}{\omega}\cos(\omega T)
\label{eq:U_z}
\end{equation}
where $\omega$ is the angular frequency and $T$ is the time.
No-slip boundary conditions are imposed on all solid boundaries, including the surface of the sphere and the walls of the box.
To non-dimensionalize the system, $D$, $U_{max}$, $U_{max}/D$, and $\rho U_{max}^2$ are utilized as characteristic scales of length, velocity, time, and pressure, respectively, yielding the following relationships governing the sphere's kinematics:
\begin{equation}
    \textit{u}_z = \sin\bigg(\frac{D}{A}t\bigg), \qquad
    \textit{z} = -\frac{A}{D}\cos\bigg(\frac{D}{A}t\bigg)
    \label{eq:u_z_and_z}
\end{equation}
where $A=U_{max}/\omega$ represents the amplitude of the oscillations.
The system's behavior is governed by two non-dimensional parameters:
the Reynolds number, $Re=U_{max}D/\nu$, and the amplitude ratio, $A/D$ (where $D/A$ is also known as the reduced frequency).
A snapshot of representative flow including pathlines and contours of the vertical velocity component $u_z$ around the sphere, typical of the given set-up, is presented in Fig. \ref{fig:flow}.

\begin{figure}
\begin{center}
	\newcommand{\image}[1]{\includegraphics[width=0.4\linewidth]{./#1}}
	\subfigure[Physical model]{\image{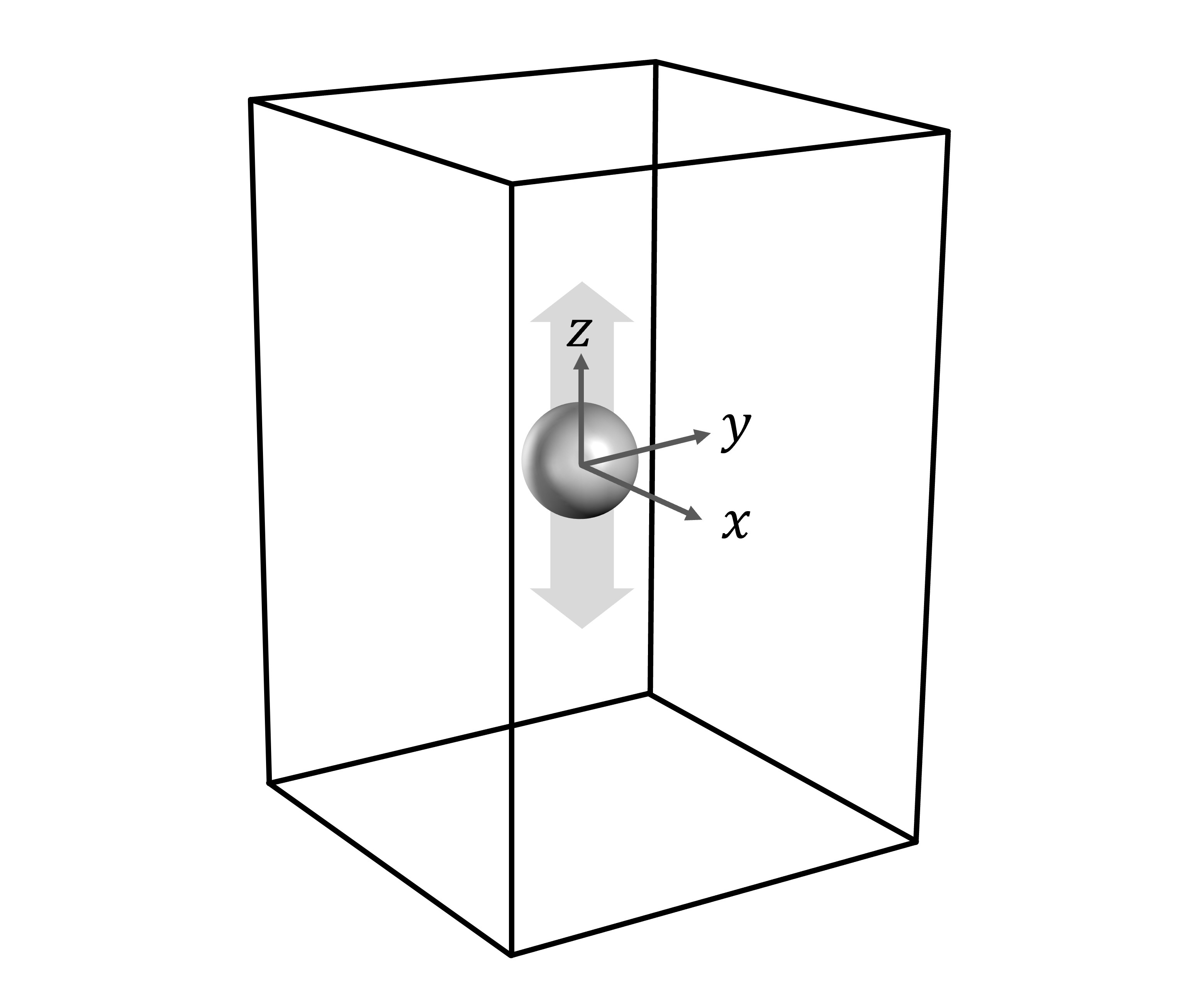}\label{fig:ball}}
	\subfigure[Snapshot of a representative flow]{\image{flow.eps}\label{fig:flow}}
    \hspace{10pt}
    \\
\end{center}
\caption{An oscillating sphere of diameter $D$ and a box of dimensions $4D\times 4D \times 6D$}
\label{fig:3D}
\end{figure}

For completeness, we recall the force balance equation of an accelerating sphere in an otherwise quiescent fluid, as described in \cite{sela2021semi}. 
In accordance with the IBM formalism where the sphere is filled with the same fluid as that outside, this equation can be expressed as \cite{breugem2012second}:
\begin{equation}
\frac{d}{dt}\int_{V_{sph}} \bfu dV = \bff_D + \int_{V_{sph}} \bff dV
\label{eq:force_bal}
\end{equation}
where $\bff_D$ is the instantaneous drag force exerted by the fluid on the sphere and the integral of $\bff$ is the instantaneous external force required to maintain the prescribed kinematics of the body. 
This external force can be directly computed by summing all the IBM forces:
\begin{equation}
\int_{V_{sph}} \bff dV \approx\sum_{i,j,k}\bff_{ijk}\Delta x \Delta y \Delta z.
\label{eq:IBM_force}
\end{equation}
The drag force, $\bff_D$, can be determined by using Eqs. \eqref{eq:force_bal} and \eqref{eq:IBM_force}. 
Following \cite{uhlmann2005JCP}, we assume rigid-body motion of the fluid within the sphere, allowing us to approximate the left-hand side term of Eq.~\eqref{eq:force_bal} as
\begin{equation}
\frac{d}{dt}\int_{V_{sph}} \bfu dV \approx \frac{d\bfu_c}{dt}V_{sph}
\label{eq:IBM_acceler_approx}
\end{equation}
where $d\bfu_c/dt$ is the acceleration of the sphere's center of mass. For comparison purposes, we express the drag force $\bff_D$ in terms of the drag coefficient for a spherical geometry:
\begin{equation}
\bfC_{D}=\frac{8\bff_D}{\rho U_{max}^2\pi D^2}.
\label{eq:Drag_Coef}
\end{equation}
After scaling the drag force by $ \rho U_{max}^2 D^2$ the $z$ component of the drag coefficient $C_D\equiv (C_D)_z$ becomes $C_{D}={8f_{D_{z}}}/{\pi}$.

To verify our method, we compare the time evolution and maximum values of the drag coefficient $C_{D}$ obtained on a $400\times400\times600$ grid for four Reynolds numbers in the range $50\leq Re\leq200$ at $A/D=1$ with the corresponding data reported in \cite{sela2021semi} and \cite{blackburn2002mass}. 
Following the principles outlined in Section 2.1, the immersed boundary is represented by Lagrangian points uniformly distributed with spacing approximately equal to the Eulerian grid size. 
This uniform distribution was achieved using Leopardi's noniterative method \cite{leopardi2006partition}, which ensures each Lagrangian point is associated with a virtual surface region of equal area. The comparison of the drag coefficient time evolution is presented in Fig. \ref{fig:drag_AtoD1}, where the results of our current simulations are shown by solid lines, while the reference data reported in \cite{sela2021semi} are denoted by filled circles. The sphere's position is indicated by a dotted line. 
Excellent agreement is obtained for the entire range of governing parameters, with the maximum relative deviation not exceeding 2.1\%, thus successfully verifying our results. 
In Fig. \ref{fig:maxDrag}, the maximum drag coefficient values obtained in the present study are compared with the results for the entire range of $Re$ and $A/D$ values reported in \cite{blackburn2002mass}, under the assumption of axisymmetric non-confined flow. 
The maximum drag coefficients for $A/D=0.5,1$ and $1.5$ are shown by filled circles, while the results from \cite{blackburn2002mass} are denoted by a solid line. 
The results demonstrate good agreement, with deviations not exceeding 6.7\%. 
The slightly higher values obtained in the present study can be attributed to the presence of walls with no-slip conditions in our setup.

\begin{figure}
\begin{center}
	\newcommand{\image}[1]{\includegraphics[width=0.47\linewidth]{./#1}}
    \subfigure[Time evolution of the drag coefficient, $C_D$]{\image{drag_AtoD1.eps}\label{fig:drag_AtoD1}}
    \hspace{10pt}
\subfigure[Maximal values of the drag coefficent, max($C_D$)]{\image{maxDrag.eps}\label{fig:maxDrag}}
    \\
\end{center}
\caption{Time evolution of the drag coefficient $C_D=8f_z/\pi$ as a function of time for $A/D=1$, and maximal values of the drag coefficient, $\max(C_D)$, as a function of $A/D$.
The dotted line shows the position of the sphere's center (right axis).}
\label{fig:drag_time}
\end{figure}

\subsubsection{Model problem 2: multiple packed spheres (one-way coupling)}
\label{subsubsec:PorousSphere}

In the previous subsection we verified our method for the case of a single body.
In this subsection we preform additional experiments to further demonstrate the capabilities of the developed method for simulation of multiple moving bodies.

\begin{figure}
\begin{center}
    \newcommand{\image}[1]{\includegraphics[width=0.25\linewidth]{./#1}}
    \subfigure[7 packed sub-spheres]{\image{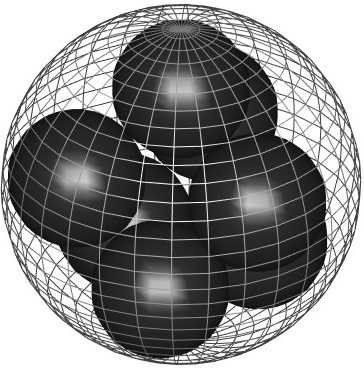}\label{fig:sevenSpheres}}
    \hspace{100pt}
    \subfigure[14 packed sub-spheres]{\image{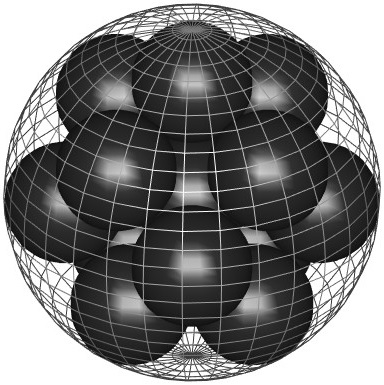}\label{fig:fourteenSpheres}}
\end{center}
\caption{Schematic representation of porous spheres modeled by arrays of (a) 7 and (b) 14 solid sub-spheres packed without overlap within a bounding sphere of diameter $D$.}
\label{fig:spheresPackaging}
\end{figure}

To this end, we focus on investigating flow properties in general and drag coefficients in particular that are typical of oscillating porous spheres.
This knowledge is critical for understanding and controlling many practical applications, including the behavior of particle clusters in fluidized bed reactors, the movement of particulate matter in power generation systems, the dynamics of microcarrier transport in bioprocess applications, and the settling behavior of sediments in coastal waters \cite{li2016rotation, yu2012numerical}, to name a few.
Following the methodology of Ma et al. \cite{ma2020particle}, a porous sphere is modelled as an array of small solid sub-spheres packed without overlap within a larger bounding sphere of diameter $D$.
The packing configuration of these small spheres is not unique and can be achieved through various approaches.
For the present study, we employed the serial symmetrical relocation algorithm developed by Huang and Liang \cite{huang2012serial} to generate configurations with 7 and 14 sub-spheres.

\begin{table}[h]
\centering
\begin{tabular}{lccc lccc}
\hline
\toprule
\multicolumn{4}{c}{7 sub-sphere configuration, $r=0.193$, $\varepsilon=0.598$} & \multicolumn{4}{c}{14 sub-sphere configuration, $r=0.161$, $\varepsilon=0.526$} \\
\midrule
 \small  \# Sub- & \multicolumn{3}{c}{Coordinates of sub-sphere center} & \small \# Sub- & \multicolumn{3}{c}{Coordinates of sub-sphere center} \\
\small sphere & $x\times 10$ & $y\times 10$ & $z\times 10$ & \small sphere & $x\times 10$ & $y\times 10$ & $z\times 10$ \\
\midrule
\small 1 & \small$1.1140$ & \small$-1.9296$ & \small$-2.1126$ & \small 1 & \small$0$ & \small$-2.287$ & \small$-2.4917$ \\
\small 2 & \small$-2.2281$ & \small$0$ & \small$-2.1126$ & \small 2 & \small$-2.2875$ & \small$0$ & \small$-2.4917$ \\
\small 3 & \small$1.1140$ & \small$1.9296$ & \small$-2.1126$ & \small 3 & $2.2875$ & \small$0$ & \small$-2.4917$ \\
\small 4 & \small$-1.5009$ & \small$-2.5997$ & \small$0.64522$ & \small 4 & \small$0$ & \small$2.2875$ & \small$-2.4917$ \\
\small 5 & \small$3.0019$ & \small$0$ & \small$0.64522$ & \small 5 & \small$-2.3820$ & \small$-2.3820$ & \small$-0.30478$ \\
\small 6 & \small$-1.5009$ & \small$2.5997$ & \small$0.64522$ & \small 6 & \small$2.3820$ & \small$-2.3820$ & \small$-0.30478$ \\
\small 7 & \small$0$ & \small$0$ & \small$3.0704$ & \small 7 & \small$-2.3820$ & \small$2.3820$ & \small$-0.30478$ \\
 &  &  &  & \small 8 & \small$2.3820$ & \small$2.3820$ & \small$-0.30478$ \\
 &  &  &  & \small 9 & \small 0 & \small 0 & \small 0 \\
 &  &  &  & \small 10 & \small$0$ & \small$-2.8412$ & \small$1.8355$ \\
 &  &  &  & \small 11 & \small$-2.8412$ & \small$0$ & \small$1.8355$ \\
 &  &  &  & \small 12 & \small$2.8412$ & \small$0$ & \small$1.8355$ \\
 &  &  &  & \small 13 & \small$0$ & \small$2.8412$ & \small$1.8355$ \\
 &  &  &  & \small 14 & \small$0$ & \small$0$ & \small$3.3825$ \\
\bottomrule
\end{tabular}
\caption{Center coordinates of 7 and 14 sub-spheres constituting the model of a bounding sphere of diameter $D$. Here $r$ is the radius of sub-spheres normalized by $D$, and $\varepsilon$ is the corresponding porosity. The center of the bounding sphere is located at the origin of the coordinate system. }
\label{tab:spheresCenters}
\end{table}

The details regarding the radii and packing arrangements for each configuration of sub-spheres are given in Table \ref{tab:spheresCenters}.
The table also includes the corresponding porosity values $\varepsilon$, defined as the ratio of void volume to total volume ($\varepsilon = V_{void}/V_{total}$).
For numerical stability in the IBM implementation, each sub-sphere's radius was reduced by half a grid spacing to ensure proper surface contact resolution while avoiding ill-conditioning.
The physical configurations of the bounding spheres, containing 7 and 14 sub-spheres, are illustrated in Figs. \ref{fig:sevenSpheres} and \ref{fig:fourteenSpheres}, respectively.
A spherical wire-frame is shown in these figures to indicate the boundary of the theoretical bounding sphere.

To understand the dynamics of the system, we analyze the external forces (calculated using Eq. \eqref{eq:IBM_force}) required to drive the oscillatory motion of each sub-sphere.
Figure \ref{fig:External} presents the temporal evolution of these external forces (left axis) for the 7 sub-sphere and 14-sphere arrays under three different Reynolds numbers: $Re=100$, $Re=200$, and $Re=300$.
The dotted line represents the temporal evolution of the $z$-coordinate of the bounding sphere's center (right axis).

\begin{figure}
\begin{center}
	\newcommand{\image}[1]{\includegraphics[width=0.47\linewidth]{./#1}}
    \subfigure[$Re=100$, 7 sub-spheres]{\image{7_Re100_External.eps}\label{fig:7_Re100_External}}
    \hspace{20pt}
    \subfigure[$Re=100$, 14 sub-spheres]{\image{14_Re100_External.eps}\label{fig:14_Re100_External}}
    \hspace{20pt}
     \subfigure[$Re=200$, 7 sub-spheres]{\image{7_Re200_External.eps}\label{fig:7_Re200_External}}
     \hspace{20pt}
     \subfigure[$Re=200$, 14 sub-spheres]{\image{14_Re200_External.eps}\label{fig:14_Re200_External}}
     \hspace{20pt}
     \subfigure[$Re=300$, 7 sub-spheres]{\image{7_Re300_External.eps}\label{fig:7_Re300_External}}
     \hspace{20pt}
     \subfigure[$Re=300$, 14 sub-spheres]{\image{14_Re300_External.eps}\label{fig:14_Re300_External}}
    \\
\end{center}
\caption{Time evolution of external forces applied to individual sub-spheres in the 7-sphere array and in the 14-sphere array during oscillatory motion. Each solid curve corresponds to the force applied to a sub-sphere identified by its serial number from Table \ref{tab:spheresCenters}.}
\label{fig:External}
\end{figure}

The analysis of Figs. \ref{fig:7_Re100_External}, \ref{fig:7_Re200_External} and \ref{fig:7_Re300_External} demonstrates that sub-spheres experience different temporal forces, with their magnitudes and patterns determined by vertical position within the bounding sphere.
Sub-spheres at equal vertical positions experience identical temporal forces regardless of their horizontal coordinates, confirming the axial nature of the force distribution.
The temporal evolution of forces on individual sub-spheres exhibits clear asymmetry, quantified by differences between positive and negative force peaks.
Sub-spheres \#1, \#2, and \#3, positioned below the bounding sphere's center, experience negative forces approximately 50\% larger than their positive counterparts during downward motion, with reduced forces during upward motion due to shielding from upper sub-spheres.
Sub-sphere \#7, positioned at the top edge of the bounding sphere, displays an inverse asymmetry pattern, with positive force peaks exceeding negative values by approximately 40\%.
Sub-spheres \#4, \#5, and \#6, located slightly above the center and attached to the outer boundary of the bounding sphere, remain unshielded by other sub-spheres throughout the oscillation cycle.
This unshielded position results in nearly symmetric forces about zero, with peak-to-peak variations not exceeding 10\%.
The magnitude of these forces systematically decreases with increasing Reynolds number across all cases, while maintaining consistent phase relationships relative to the bounding sphere's position.
This reduction in force amplitude demonstrates the diminishing effect of viscous forces at higher Reynolds numbers, though the fundamental force distribution patterns persist.

The behavior of external forces in the case of 14 sub-spheres exhibits similar patterns while revealing additional features due to the more complex internal structure.
The analysis of Figs. \ref{fig:14_Re100_External}, \ref{fig:14_Re200_External} and \ref{fig:14_Re300_External} confirms that sub-spheres at equal vertical positions experience identical temporal forces regardless of their horizontal coordinates.
Sub-spheres \#1 through \#4, positioned at $z \approx -0.249$ below the bounding sphere's center, experience negative forces approximately 45\% larger than their positive counterparts during downward motion, with reduced forces during upward motion due to shielding from upper sub-spheres.
Sub-spheres \#5 through \#8, located near the equatorial plane at $z \approx -0.030$, show intermediate behavior with force asymmetry less pronounced than for the bottom spheres. Their position, slightly below the geometric center, results in force patterns that reflect partial shielding effects from both upper and lower neighboring spheres.
The force distribution of sphere \#9, uniquely positioned at the exact geometric center ($z = 0$) of the confining sphere, demonstrates nearly symmetric behavior about zero, with peak-to-peak variations not exceeding 5\%, due to balanced shielding from surrounding sub-spheres.
Sub-spheres \#10 through \#13, located at $z \approx 0.184$, display moderate force asymmetry with positive peaks approximately 30\% larger than negative ones. Sub-sphere \#14, positioned at $z \approx 0.338$ at the top edge of the bounding sphere, exhibits the strongest asymmetry in the upper region, with positive force peaks exceeding negative values by approximately 35\%. As in the previous case, these asymmetric force patterns persist across all examined $Re$ values, with diminishing magnitudes at higher Reynolds numbers.

\begin{figure}
\begin{center}
	\newcommand{\image}[1]{\includegraphics[width=0.45\linewidth]{./#1}}
    \subfigure[7 sub-spheres]{\image{7_DragCoefficient.eps}\label{fig:7_DragCoefficient}}
    \hspace{20pt}
     \subfigure[14 sub-spheres]{\image{14_DragCoefficient.eps}\label{fig:14_DragCoefficient}}
    \\
\end{center}
\caption{Time evolution of the total drag coefficient, $C_D$ obtained for the values of $Re=100$, $200$ and $300$. }
\label{fig:DragCoefficient}
\end{figure}

We next analyze the total drag coefficients, $C_D$, for arrays consisting of 7 and 14 sub-spheres.
These coefficients are calculated by summing the drag forces in the $z$ direction acting on each sub-sphere within the corresponding array and multiplying by a factor of $8/\pi$ (see Eq. \eqref{eq:Drag_Coef}).
The time evolution of $C_D$, calculated for arrays consisting of 7 and 14 sub-spheres, is shown in Fig. \ref{fig:DragCoefficient}.

\begin{figure}
\centering
    \subfigure[Isosurfaces of the $Q=0.05$ criterion for $Re=100$, calculated at different points of the sphere's trajectory: the lowest point, mid-point on the way up, highest point and mid-point on the way down (from left to right)]{
        \includegraphics[width=0.22\textwidth,clip=]{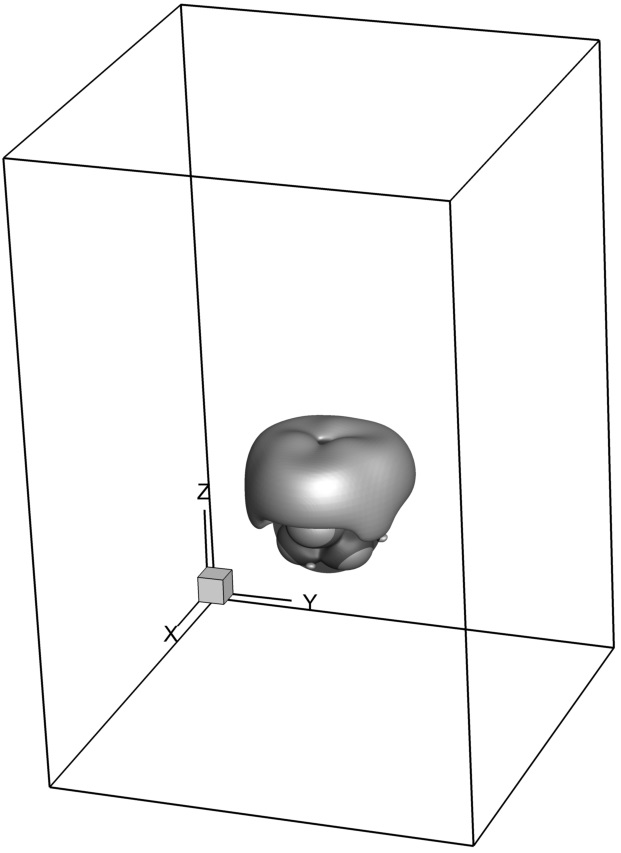}
        \hspace{10pt}
        \includegraphics[width=0.22\textwidth,clip=]{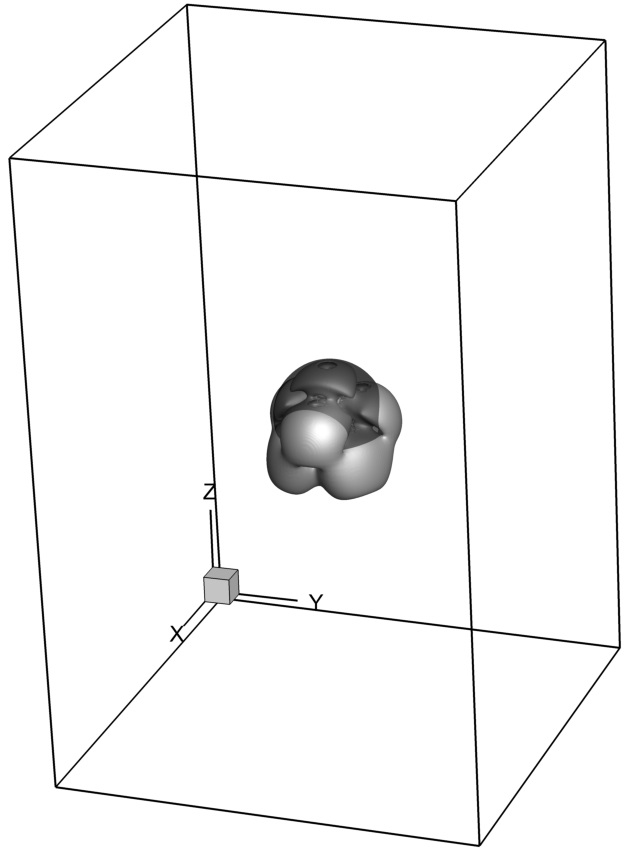}
        \hspace{10pt}
        \includegraphics[width=0.22\textwidth,clip=]{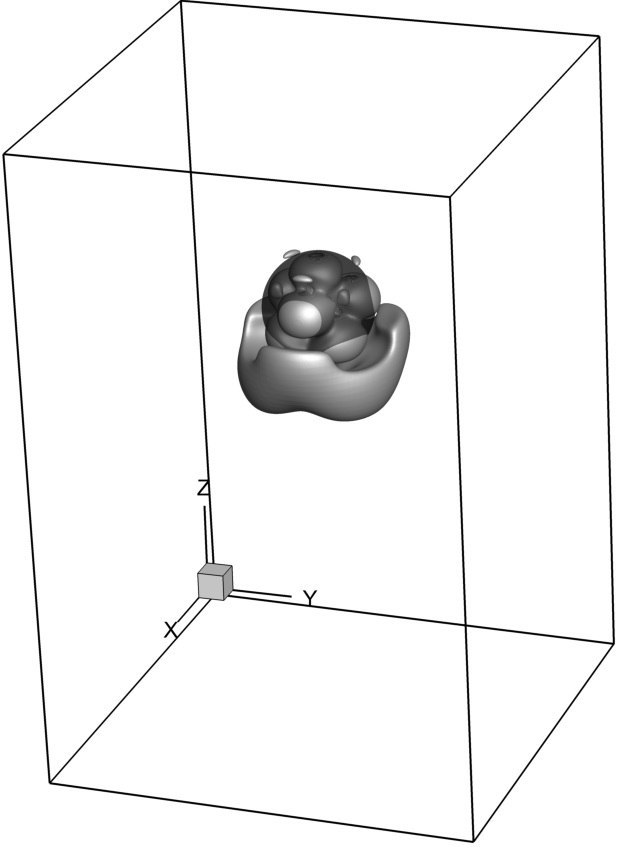}
        \hspace{10pt}
        \includegraphics[width=0.22\textwidth,clip=]{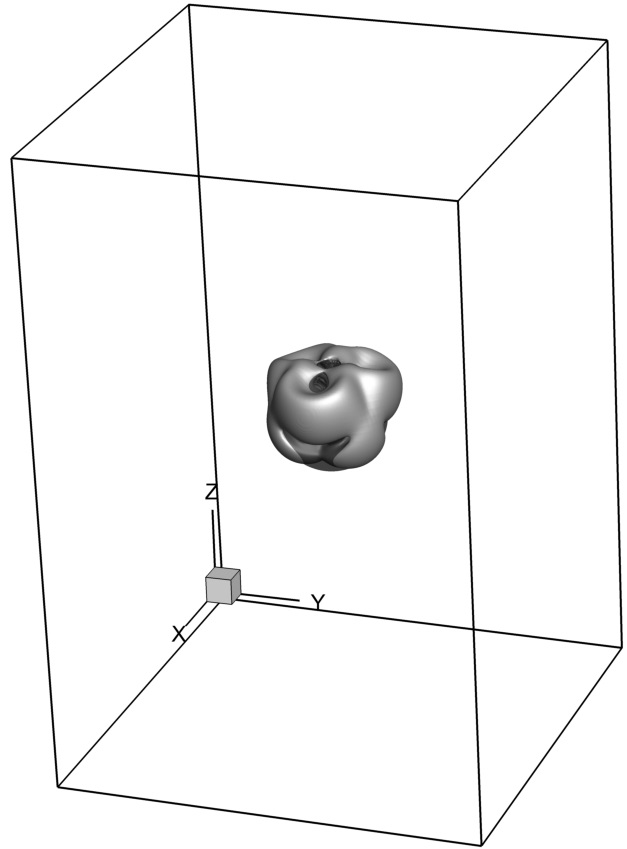}
    }
    \subfigure[Isosurfaces of the $Q=0.05$ criterion for $Re=200$, calculated at different points of the sphere's trajectory: the lowest point, mid-point on the way up, highest point and mid-point on the way down (from left to right)]{
        \includegraphics[width=0.22\textwidth,clip=]{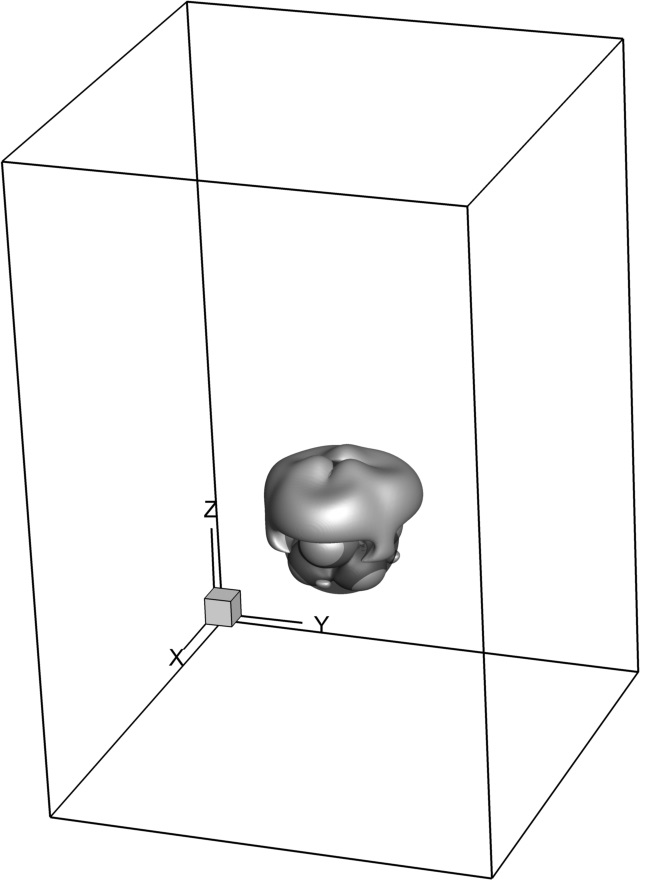}
        \hspace{10pt}
        \includegraphics[width=0.22\textwidth,clip=]{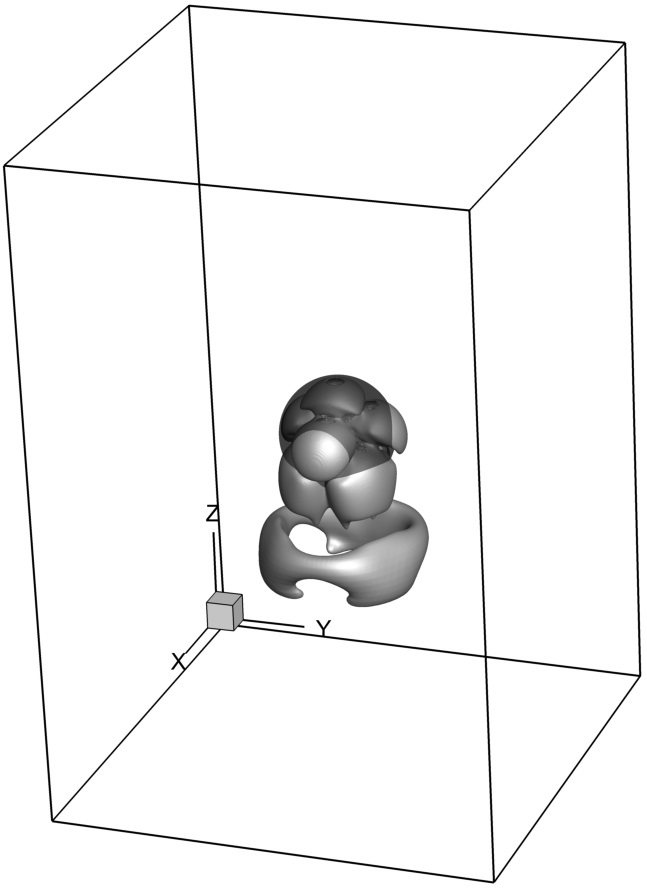}
        \hspace{10pt}
        \includegraphics[width=0.22\textwidth,clip=]{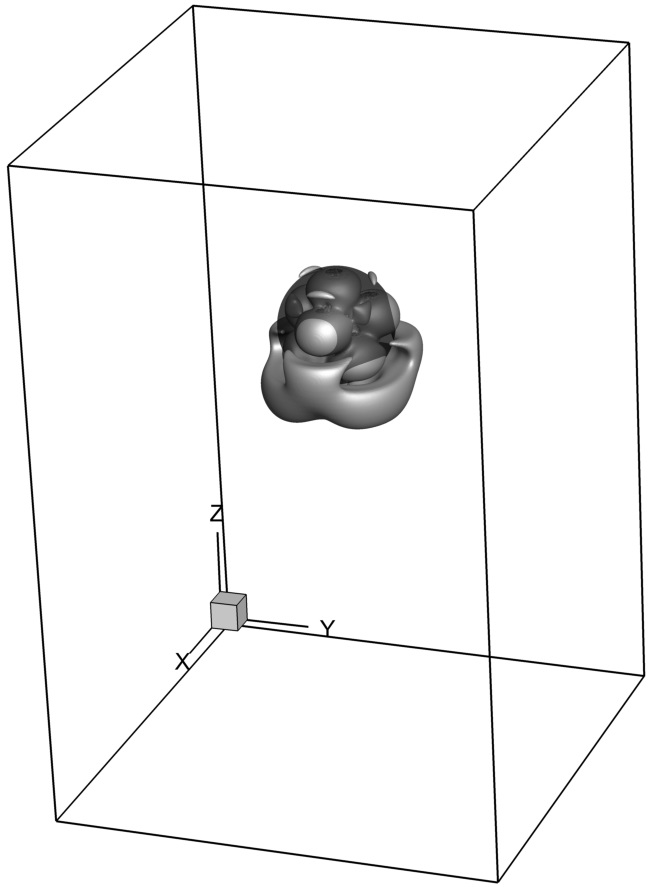}
        \hspace{10pt}
        \includegraphics[width=0.22\textwidth,clip=]{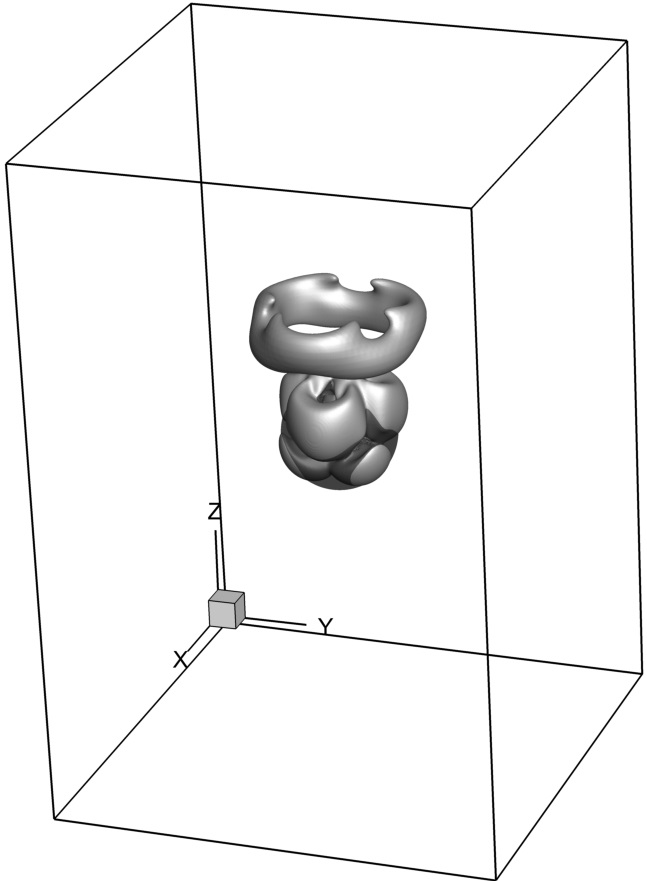}
    }
    \subfigure[Isosurfaces of the $Q=0.05$ criterion for $Re=300$, calculated at different points of the sphere's trajectory: the lowest point, mid-point on the way up, highest point and mid-point on the way down (from left to right)]{
        \includegraphics[width=0.22\textwidth,clip=]{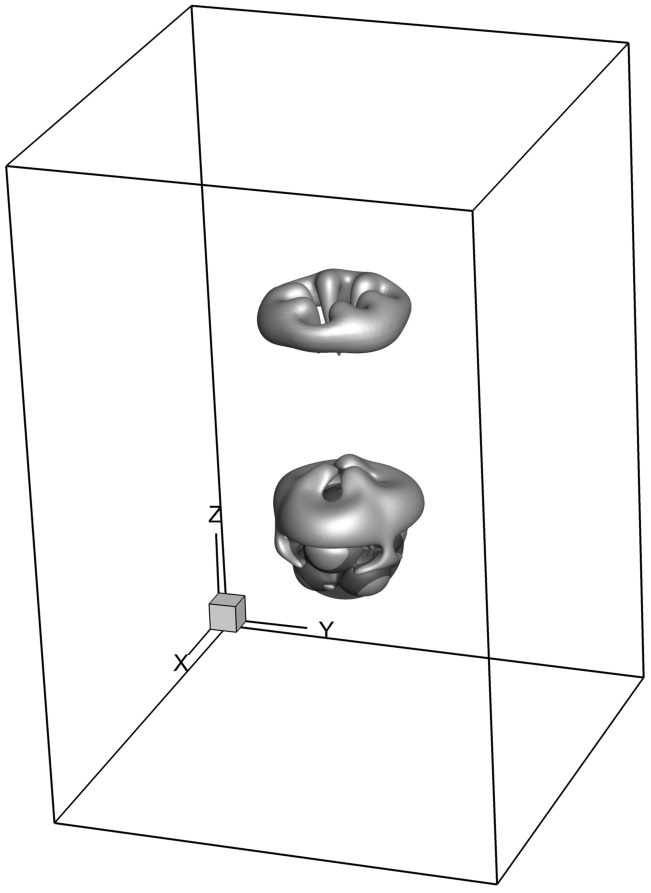}
        \hspace{10pt}
        \includegraphics[width=0.22\textwidth,clip=]{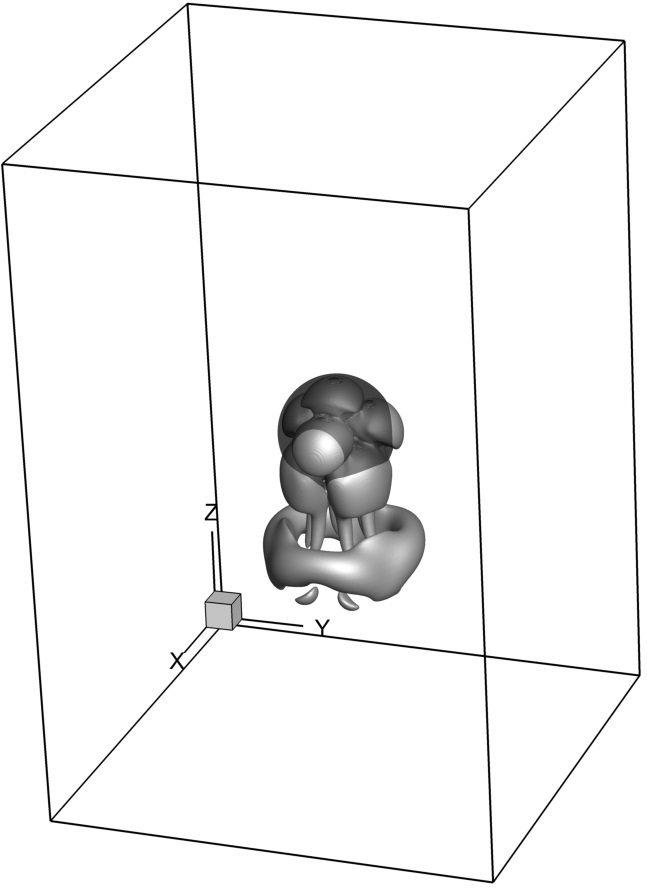}
        \hspace{10pt}
        \includegraphics[width=0.22\textwidth,clip=]{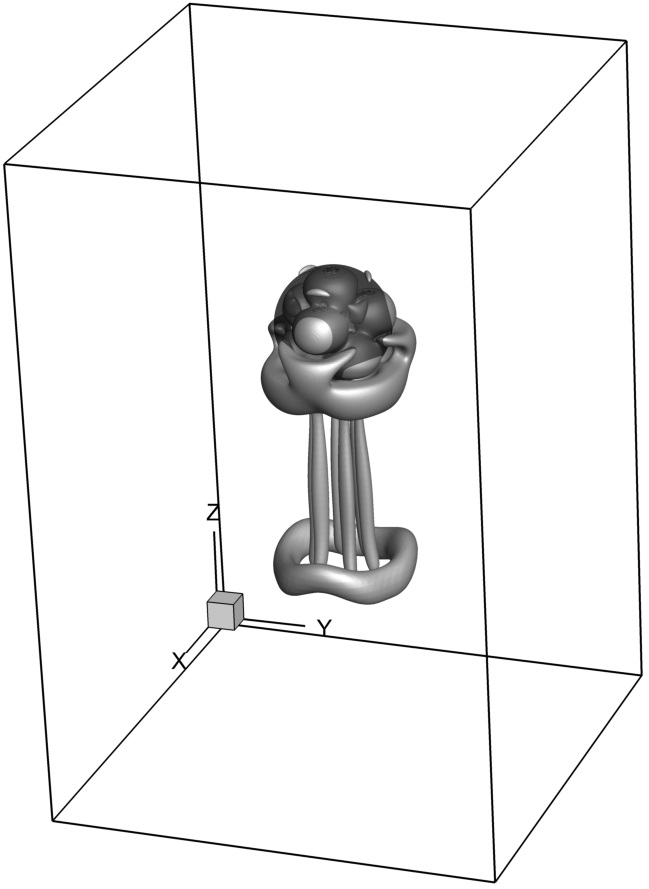}
        \hspace{10pt}
        \includegraphics[width=0.22\textwidth,clip=]{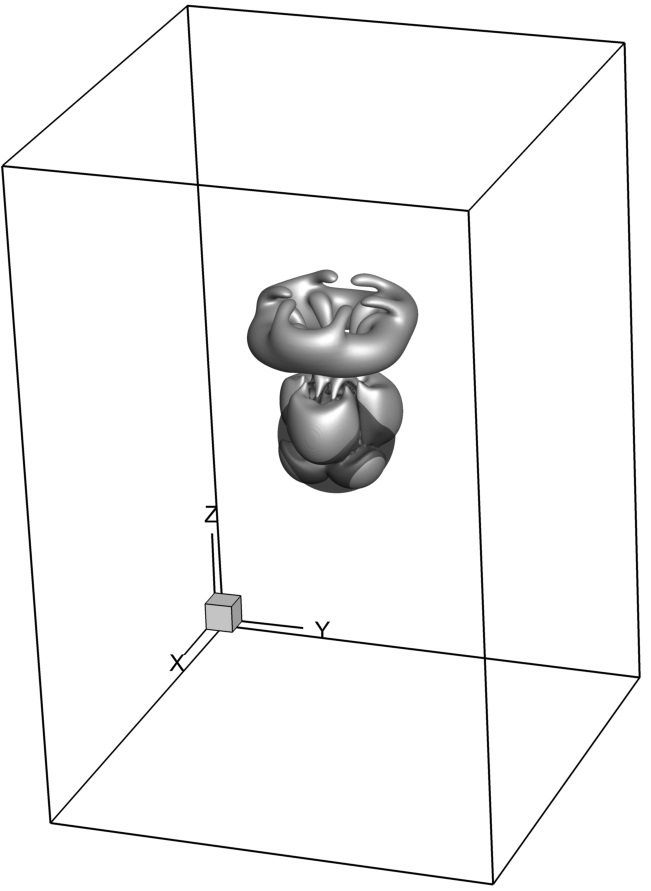}
    }
    \caption{Visualization of vortical structures (isosurfaces of $Q=0.05$) generated by a transversely oscillating sphere over a single oscillation period, calculated for $Re=100$, $200$, and $300$ at $A/D=1$. The bounding sphere containing 7 sub-spheres is indicated by a semitransparent surface.}
    \label{fig:Q_7_Sub-Spheres}
\end{figure}

Looking at Fig. \ref{fig:DragCoefficient}, we observe three key trends in the drag coefficient evolution.
First, the magnitude of $C_D$ systematically decreases with increasing Reynolds numbers, for both 7 and 14 sub-sphere arrays.
This trend directly reflects diminishing viscous effects at higher $Re$ values, resulting in reduced drag forces.
Second, the 14 sub-sphere configuration exhibits consistently higher $C_D$ values compared to the 7 sub-sphere array across all Reynolds numbers, approaching the drag coefficient values that characterize a non-porous sphere, as reported in our previous work (see \cite{sela2021semi}, Fig. 7).
This behavior is expected, as increasing the number of sub-spheres reduces the array's porosity, making it more similar to a solid sphere.
Third, a persistent phase lag exists between the sphere position and drag coefficient evolution.
This phase-lag arises from the fundamental relationship between force and motion in oscillatory flows:
the drag force is primarily dependent on the velocity rather than the displacement.
At low $Re$ values, the force reaches its maximum when the velocity is highest (at zero displacement) and its minimum when the velocity is zero (at maximum displacement), creating an inherent quarter-period phase shift.
This kinematic relationship persists across all $Re$ values, with additional inertial effects becoming more pronounced at higher Reynolds numbers but not altering the basic phase relationship.
Note that a similar phase-lag is present in Fig. \ref{fig:drag_time} from the previous subsection, corresponding to a single oscillating sphere.

\begin{figure}
\centering
    \subfigure[Isosurfaces of the $Q=0.05$ criterion for $Re=100$, calculated at different points of the sphere's trajectory: the lowest point, mid-point on the way up, highest point and mid-point on the way down (from left to right)]{
        \includegraphics[width=0.22\textwidth,clip=]{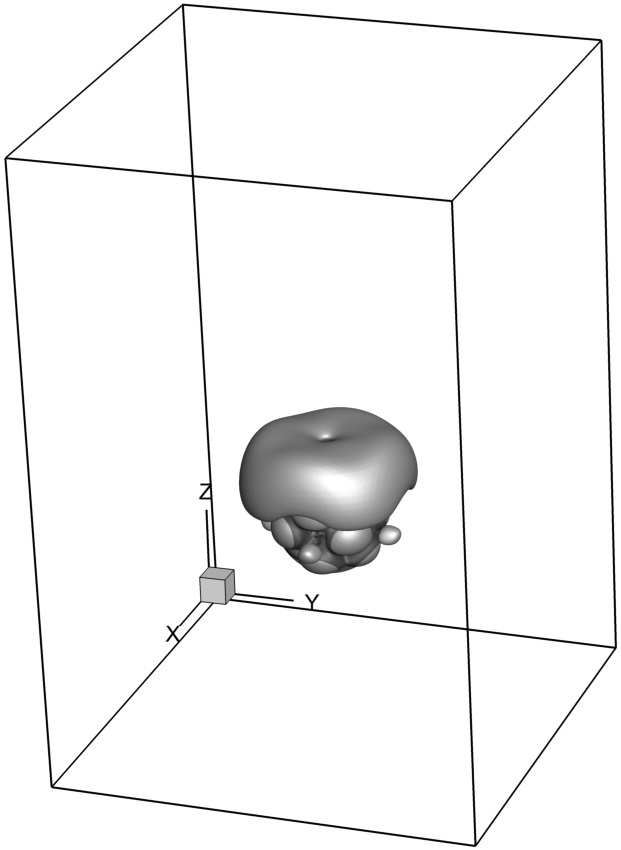}
        \hspace{10pt}
        \includegraphics[width=0.22\textwidth,clip=]{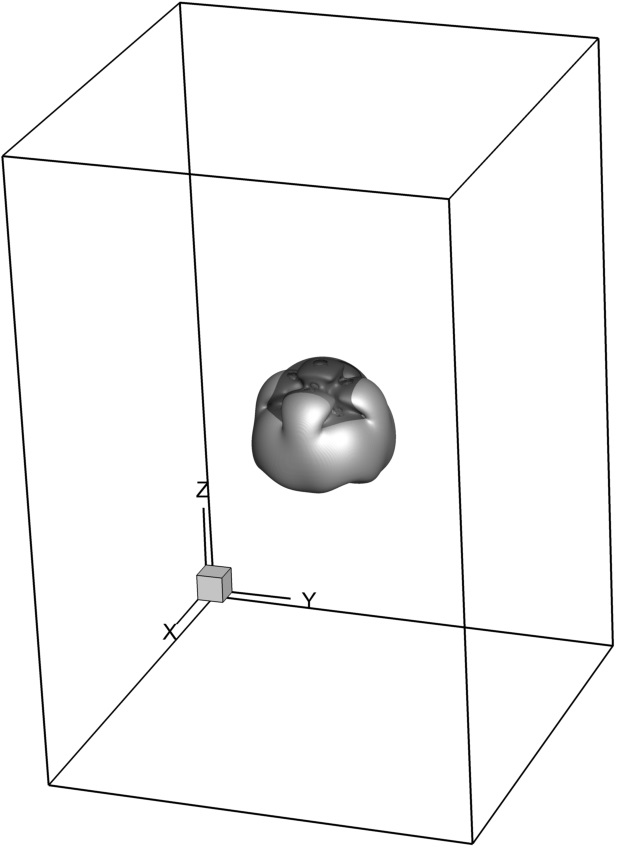}
        \hspace{10pt}
        \includegraphics[width=0.22\textwidth,clip=]{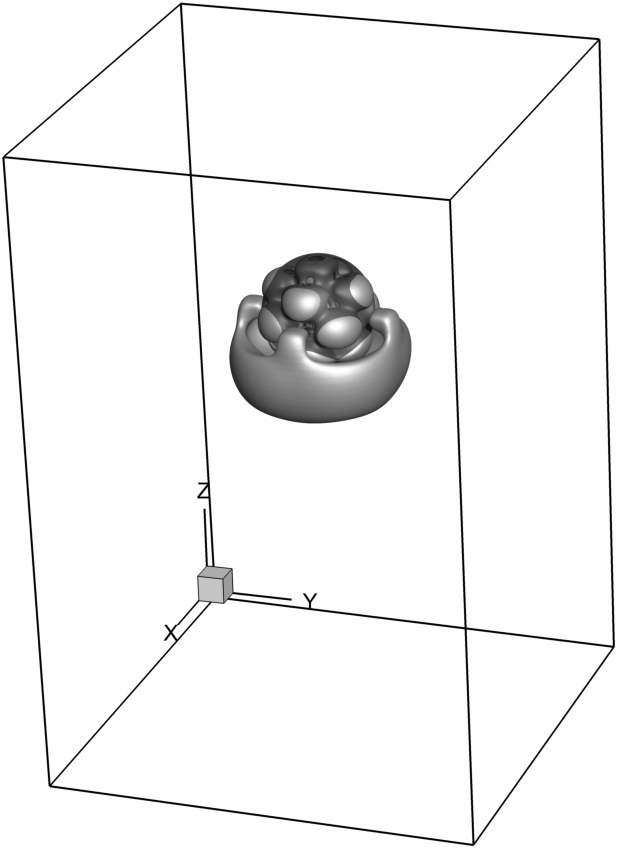}
        \hspace{10pt}
        \includegraphics[width=0.22\textwidth,clip=]{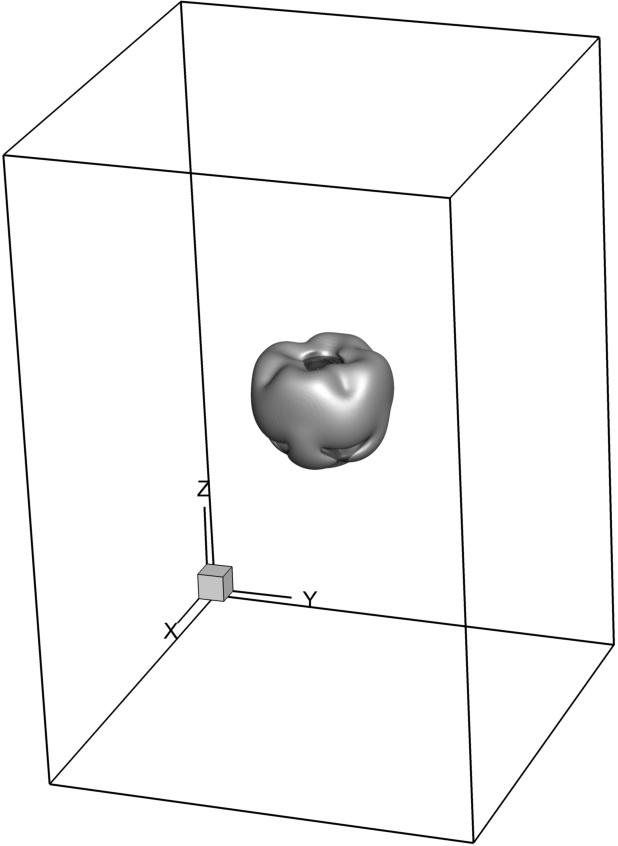}
    }
    \subfigure[Isosurfaces of the $Q=0.05$ criterion for $Re=200$, calculated at different points of the sphere's trajectory: the lowest point, mid-point on the way up, highest point and mid-point on the way down (from left to right)]{
        \includegraphics[width=0.22\textwidth,clip=]{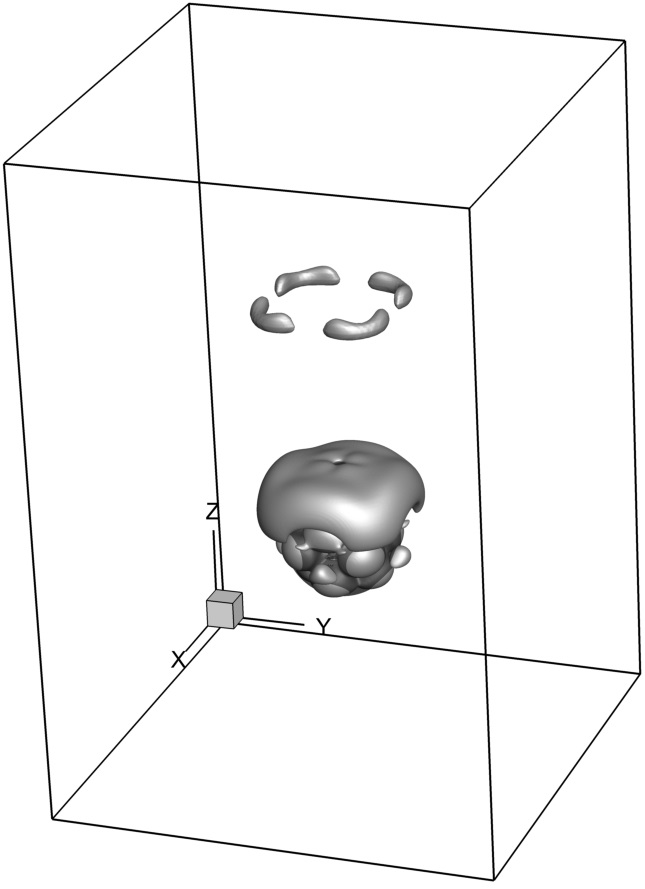}
        \hspace{10pt}
        \includegraphics[width=0.22\textwidth,clip=]{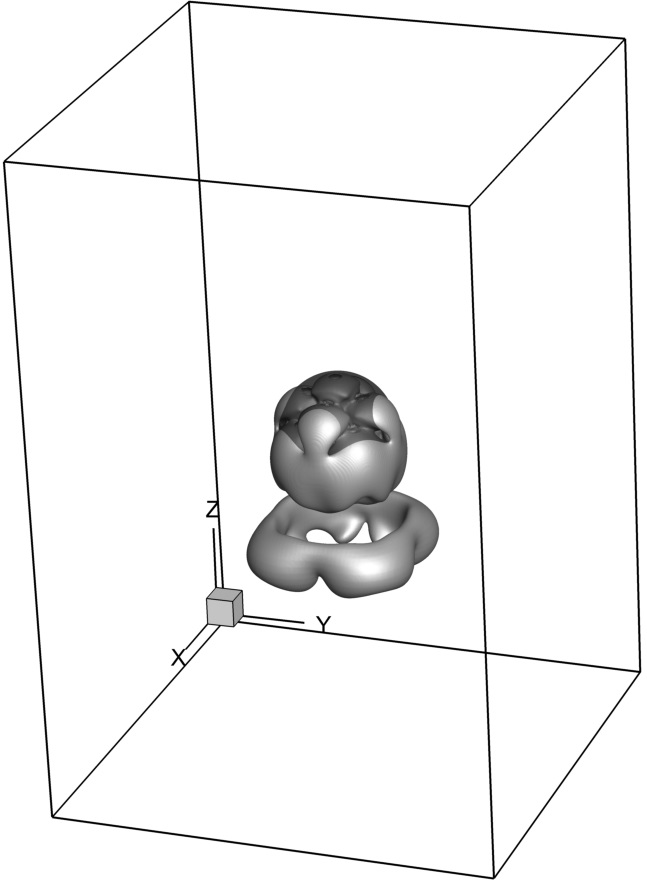}
        \hspace{10pt}
        \includegraphics[width=0.22\textwidth,clip=]{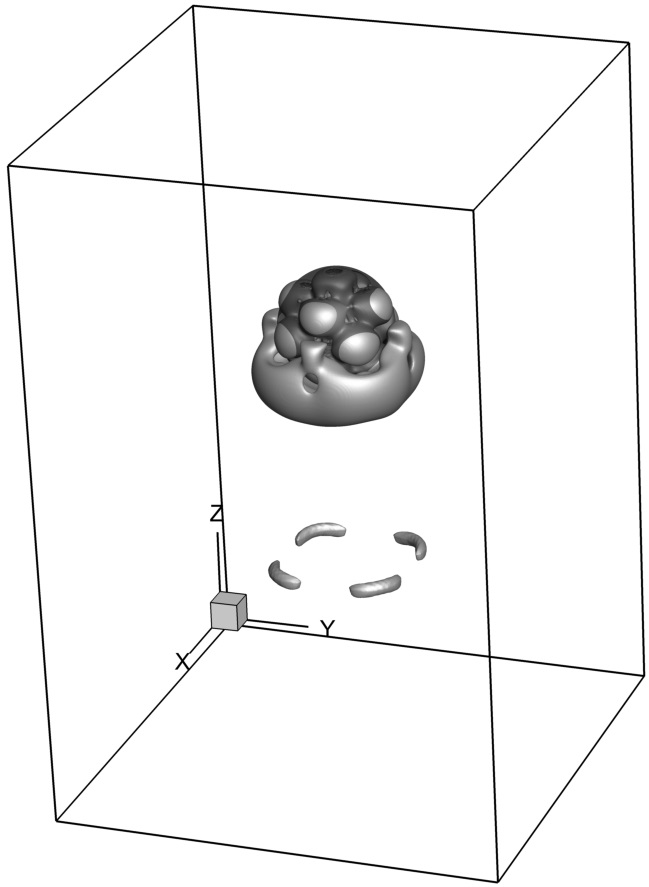}
        \hspace{10pt}
        \includegraphics[width=0.22\textwidth,clip=]{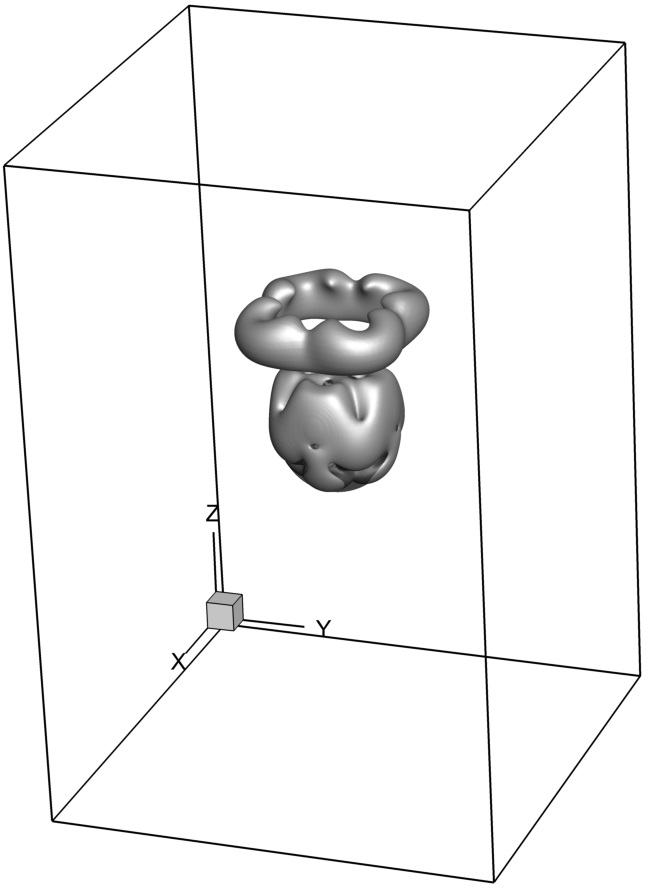}
    }
    \subfigure[Isosurfaces of the $Q=0.05$ criterion for $Re=300$, calculated at different points of the sphere's trajectory: the lowest point, mid-point on the way up, highest point and mid-point on the way down (from left to right)]{
        \includegraphics[width=0.22\textwidth,clip=]{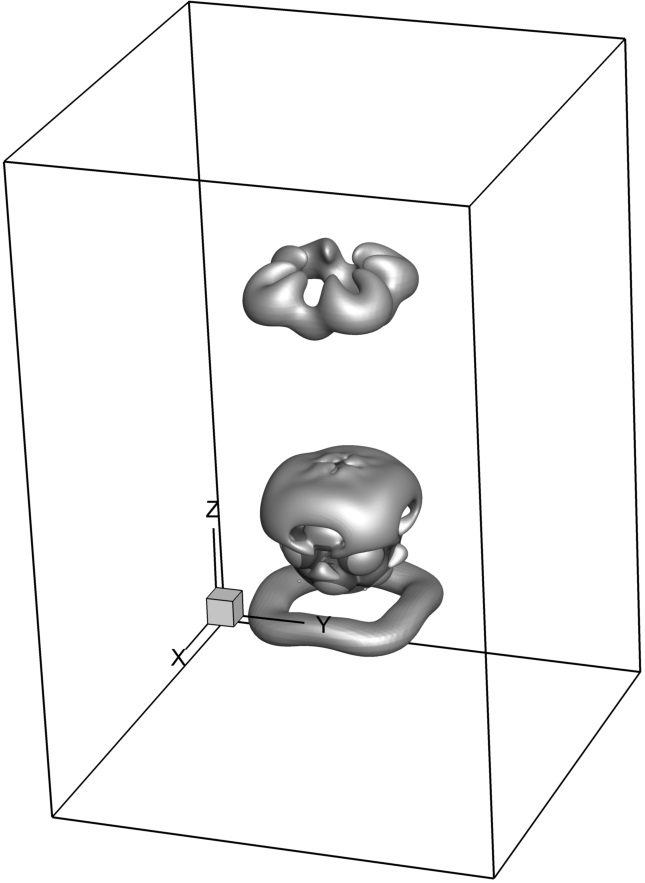}
        \hspace{10pt}
        \includegraphics[width=0.22\textwidth,clip=]{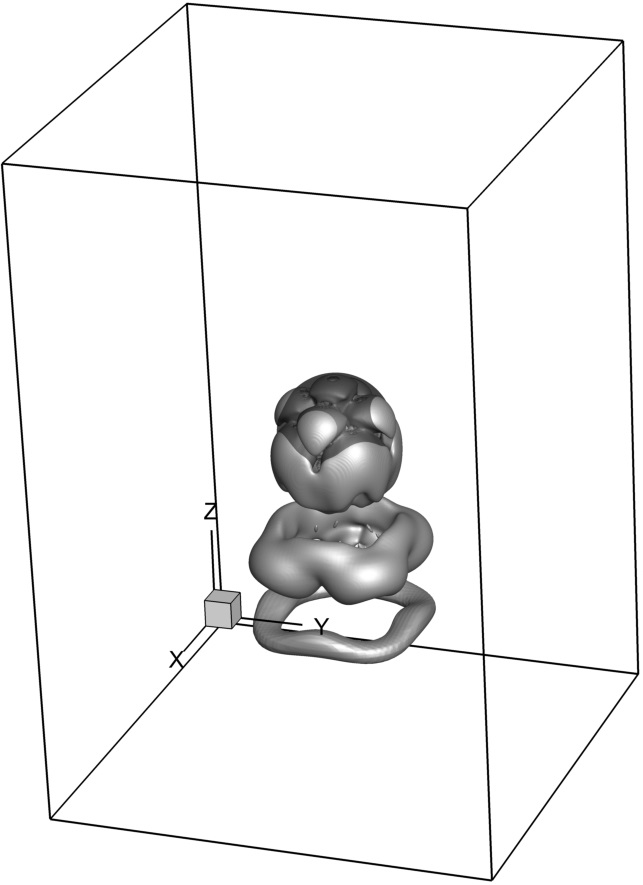}
        \hspace{10pt}
        \includegraphics[width=0.22\textwidth,clip=]{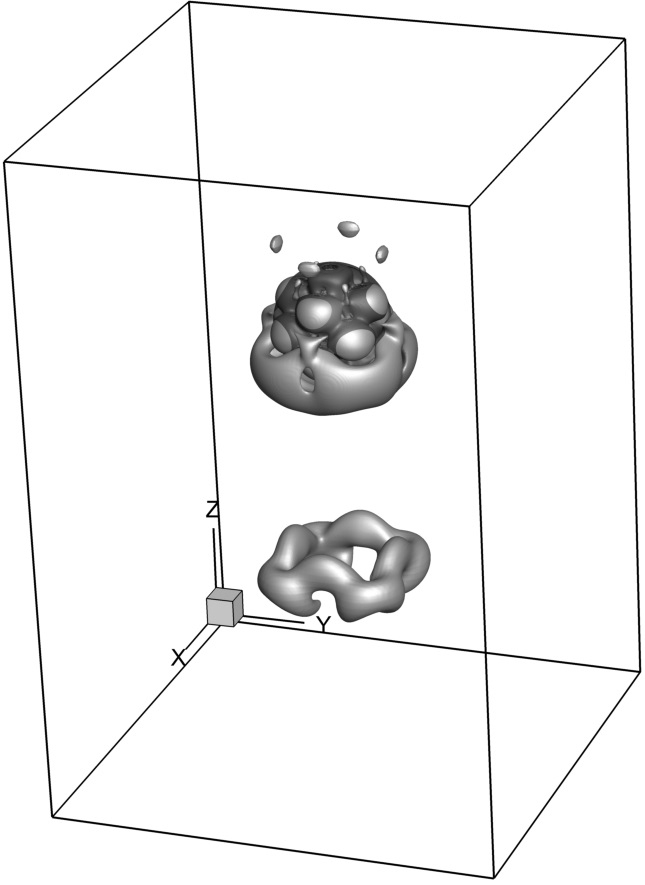}
        \hspace{10pt}
        \includegraphics[width=0.22\textwidth,clip=]{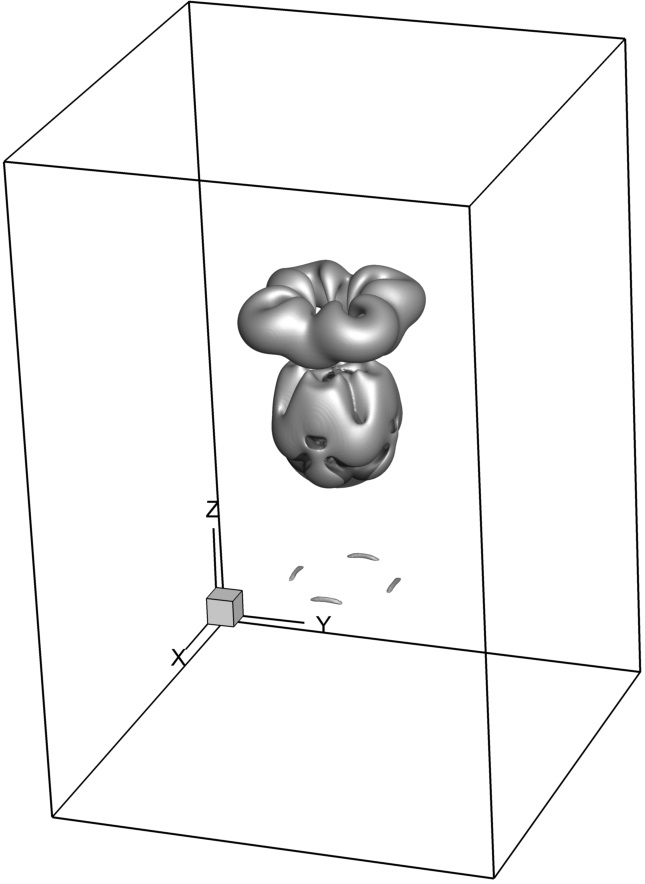}
    }
    \caption{Visualization of vortical structures (isosurfaces of $Q=0.05$) generated by a transversely oscillating sphere over a single oscillation period, calculated for $Re=100$, $200$, and $300$ at $A/D=1$. The bounding sphere containing 14 sub-spheres is indicated by a semitransparent surface.}
    \label{fig:Q_14_Sub-Spheres}
\end{figure}

To gain deeper insight into the fluid dynamics associated with the transverse motion of porous spheres, we visualize the vortical structures generated by the flow.
The visualization utilizes isosurfaces of the $Q$-criterion ($Q=0.05$), where $Q$ represents the second invariant of the velocity gradient tensor.
This value of $Q$ was chosen to effectively capture coherent vortical structures while filtering out weaker vorticity regions, thereby highlighting the dominant flow features.
The evolution of vortical structures over one oscillation period is shown in Figs. \ref{fig:Q_7_Sub-Spheres} and \ref{fig:Q_14_Sub-Spheres} for configurations containing 7 and 14 sub-spheres, respectively.
In both configurations, a clear Reynolds number dependence is observed in the persistence of vortical structures.
At $Re=100$, vortices dissipate rapidly due to dominant viscous diffusion.
As  $Re=300$, vortical structures maintain their coherence over longer periods and travel further from the sphere surface before dissipating, reflecting the diminishing role of viscous effects relative to inertial forces.

A notable distinction between the two configurations lies in their flow penetration characteristics.
The 7 sub-sphere array, with its higher porosity, exhibits more pronounced flow penetration between sub-spheres, manifested in elongated vortical structures that extend through the array (Fig. \ref{fig:Q_7_Sub-Spheres}).
In contrast, the 14 sub-sphere configuration shows more compact vortical structures that primarily develop around the array's outer boundary (Fig. \ref{fig:Q_14_Sub-Spheres}), consistent with its lower porosity.
This difference is particularly evident at $Re=300$, where the vortical structures in the 7 sub-sphere case show distinct columnar formations that are absent in the denser 14 sub-sphere arrangement.
Additionally, both configurations demonstrate asymmetric vortex-shedding patterns between upward and downward motions, with more intense vortex formation occurring during direction reversal at the trajectory extrema.
This asymmetry becomes more pronounced with increasing Reynolds number, particularly visible in the higher resolution structures at $Re=300$.

\subsubsection{Model problem 3: sedimentation and buoyant rise of a spherical particle (two-way coupling)}
\label{subsubsec:Sedimentation}

To further validate the accuracy and robustness of our numerical solver, we now consider fully two-way coupled FSI problems involving the sedimentation and buoyancy-driven motion of spherical particles in a quiescent fluid. These scenarios serve as a natural extension to the one-way coupled validation cases previously presented and pose additional numerical challenges due to the strong mutual interaction between the fluid and the solid dynamics.
In our first numerical experiment, we simulate the free sedimentation of a spherical particle in a rectangular domain of size $10D \times 10D \times 15D$, where $D$ denotes the particle diameter. The particle is initially positioned along the vertical centerline at a height of $8.5D$ and released from rest. Free-slip boundary conditions were imposed on all faces of the computational domain to minimize wall effects and allow the particle to evolve as in an unbounded domain. Two grid resolutions, $200 \times 200 \times 300$ and $400 \times 400 \times 600$, were tested to assess sensitivity to spatial discretization. As shown in Fig.~\ref{fig:vel_coord_time}, the discrepancy between the solutions obtained on the coarse and fine grids does not exceed 7\%, indicating satisfactory grid convergence. The simulated particle trajectories and velocities compare favorably with the experimental results reported by ten Cate et al.~\cite{tenCate2002settling}.
Since the fluid and solid dynamics are fully coupled, the simulations were carried out using the internal predictor-corrector iteration scheme described in Section~\ref{sec:internal_iterations}, which is required to enforce the no-slip and divergence-free constraints at the fluid-solid interface. The iteration continues until the relative $\ell_1$-norm of the incremental force-density correction falls below a prescribed threshold, typically $\varepsilon = 10^{-3}$. For all cases reported in this section, convergence was achieved within $2$ to $3$ iterations per time step. Importantly, the tested cases correspond to near-neutral buoyancy conditions, with density ratios $\rho_p/\rho_f \approx 1$, which are challenging due to the breakdown of the assumption that the fluid within the immersed body moves as a rigid body. In such regimes, the added-mass contributions must be explicitly captured for accurate prediction, further validating the correctness of our formulation.

\begin{figure}[htbp]
\centering
\subfigure[Time evolution of the vertical velocity]{
    \includegraphics[trim = 2cm 5cm 2cm 5cm, clip, width=0.45\textwidth]{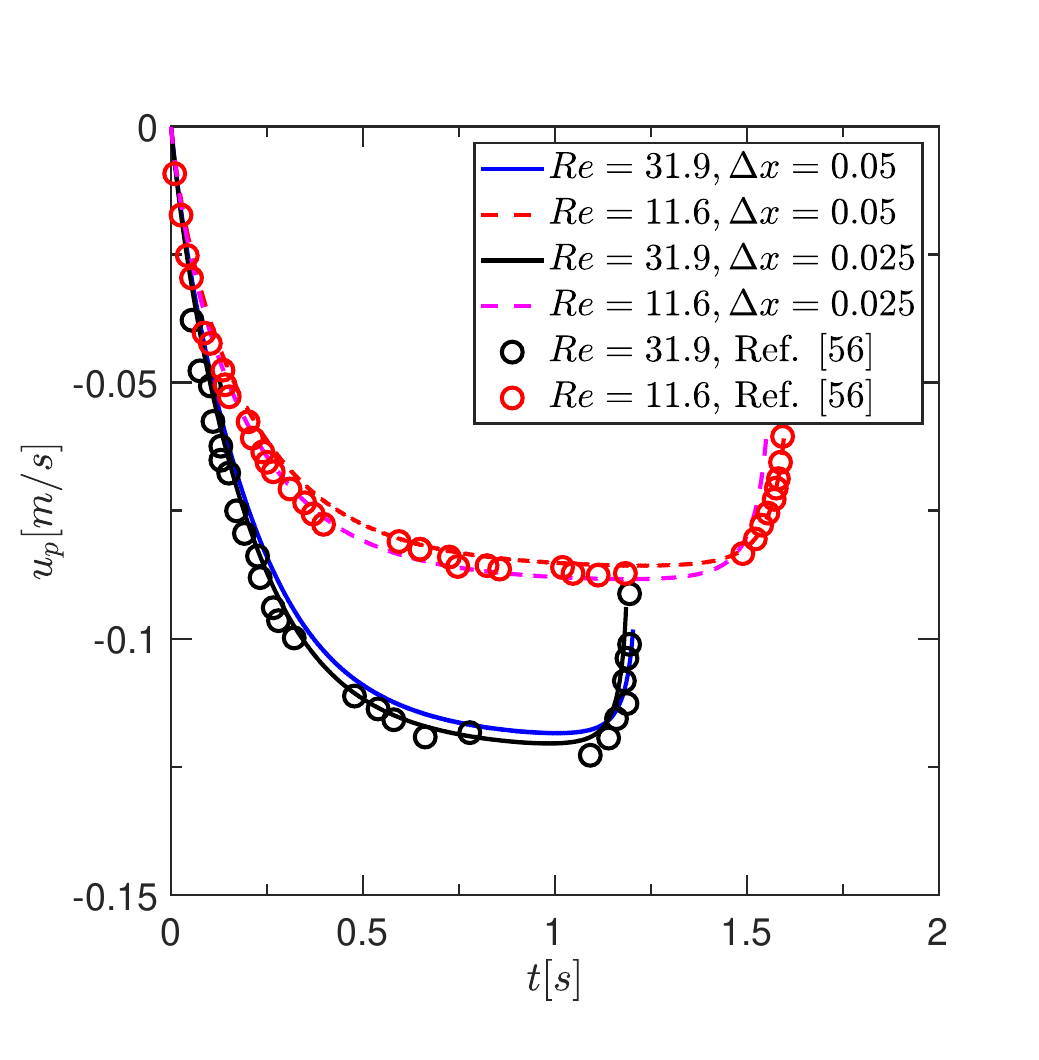}
    \label{fig:velocity_evol}
}
%\hspace{10pt}
\subfigure[Time evolution of the vertical coordinate]{
    \includegraphics[trim = 2.1cm 5.2cm 2.1cm 6cm, clip, width=0.45\textwidth]{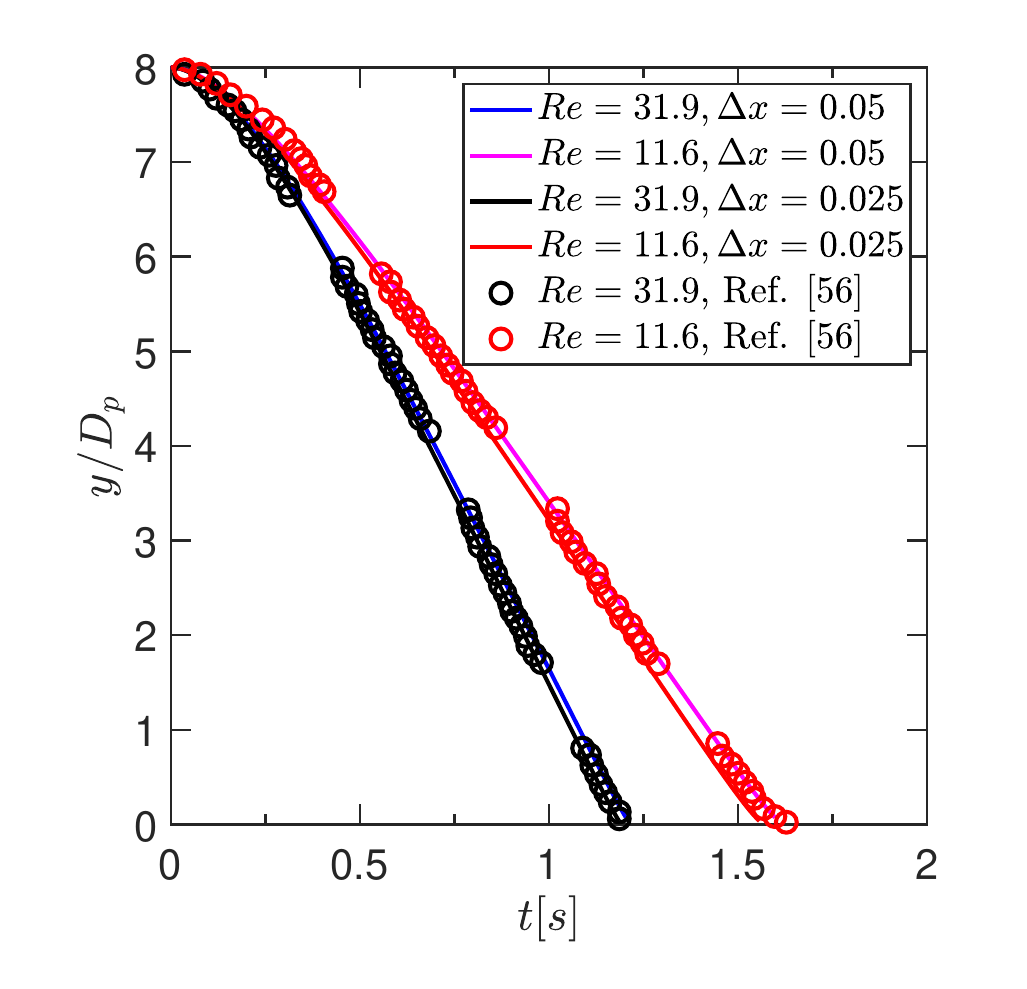}
    \label{fig:vert_coord_evol}
}
\caption{Sedimentation of a spherical particle within a quiescent fluid: present results for $\rho_p/\rho_f = 1.164$, $Re_p = 11.6$, $Fr = 0.237$ and $\rho_p/\rho_f = 1.167$, $Re_p = 31.9$, $Fr = 0.334$, compared against the experimental data of \cite{tenCate2002settling}.}
\label{fig:vel_coord_time}
\end{figure}

We further extend our validation by simulating additional two-way coupled cases involving both sedimenting and buoyant spherical particles. The computed vertical velocities are normalized using $u_{\mathrm{ref}} = \sqrt{gD_p}$ and $t_{\mathrm{ref}} = \sqrt{D_p/g}$ for velocity and time, respectively. The results are compared against both experimental data~\cite{mordant2000velocity} and previous numerical simulations~\cite{kempe2012improved}, as shown in Fig.~\ref{fig:coord1_coord2_time}. Excellent agreement is observed across entire range of density ratios and Reynolds numbers, further successfully verifying  our two-way coupled immersed boundary solver.

\begin{figure}[htbp]
\centering
\subfigure[Present results (solid lines); experimental data from \cite{mordant2000velocity} (pluses); simulation results from \cite{kempe2012improved} (circles). Case: $\rho_p/\rho_f=2.56$, $Re_p=367$, $Fr=1.8$.]{
    \includegraphics[width=0.45\textwidth]{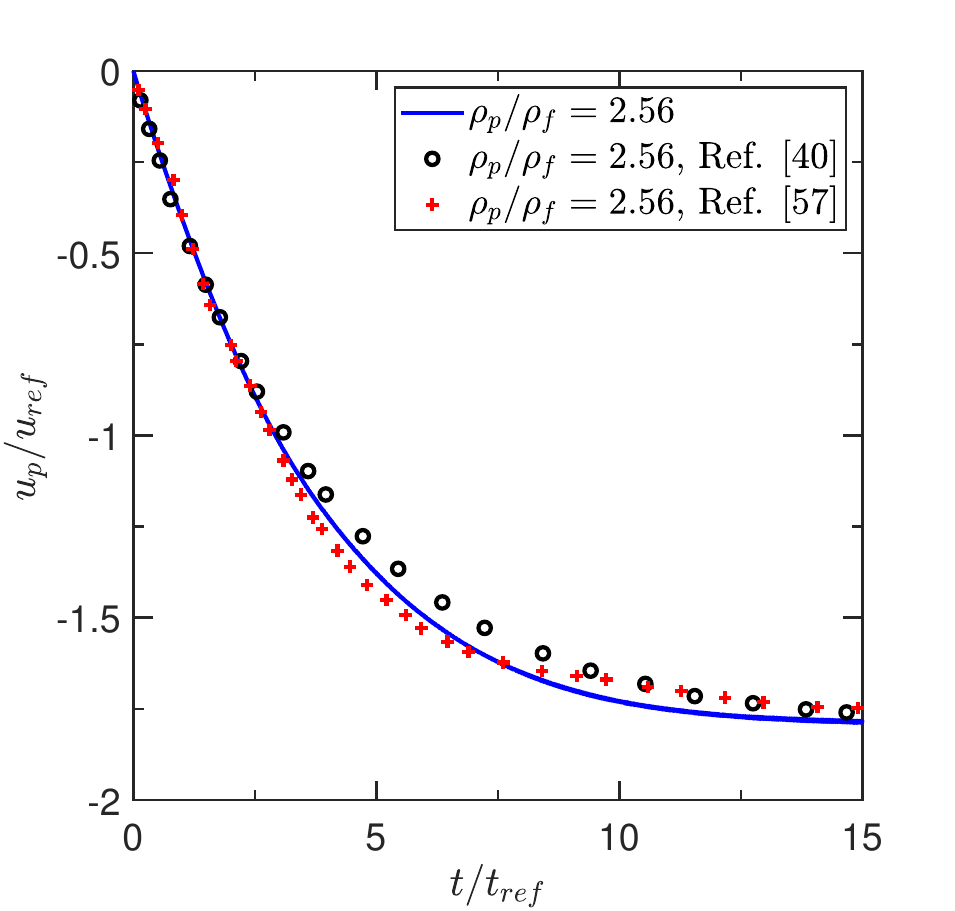}
    \label{fig:vert_vel_evol_1}
}
\hspace{10pt}
\subfigure[Present results (solid and dashed lines) compared to numerical results from \cite{kempe2012improved}. Case: $Re_p=367$, $Fr=1.8$.]{
    \includegraphics[width=0.45\textwidth]{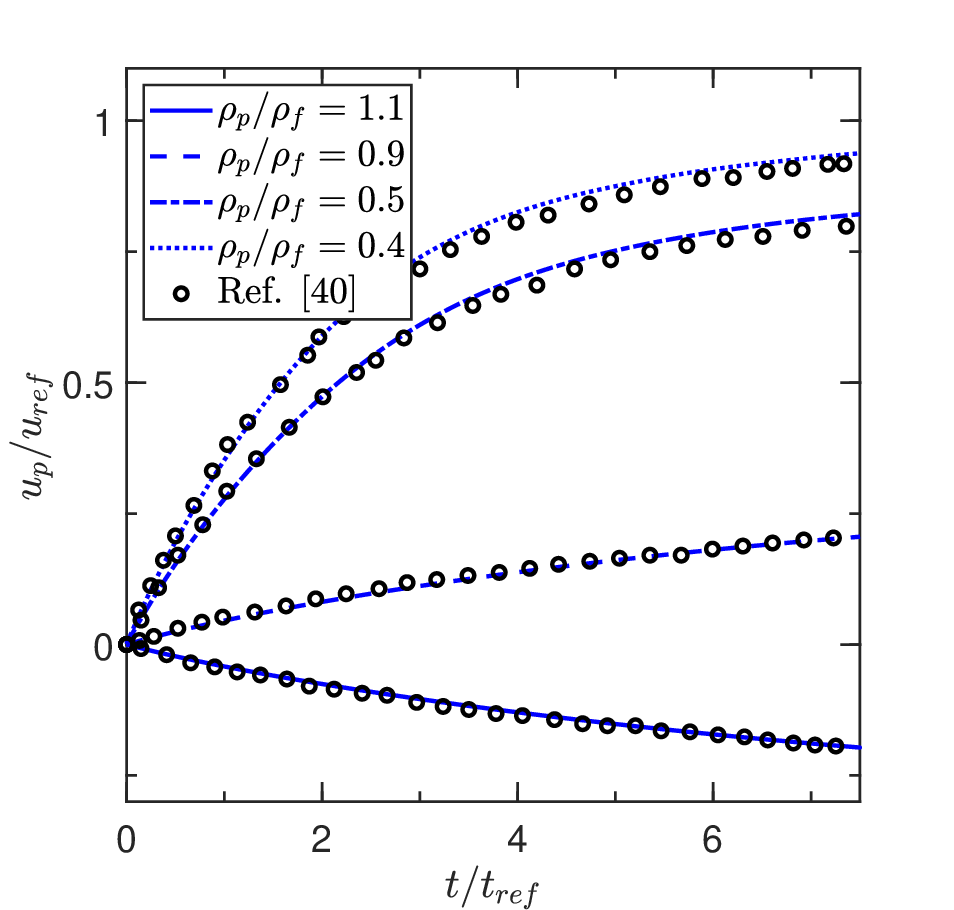}
    \label{fig:vert_vel_evol_2}
}
\caption{Time evolution of the vertical velocity of sedimenting and buoyant spherical particles in quiescent fluid: comparison with prior experimental and numerical data.}
\label{fig:coord1_coord2_time}
\end{figure}

\subsection{Preconditioner efficiency}
\label{subsec:efficiency}

The aim of this subsection is to validate the theoretical result of Theorem \ref{thm:thm} numerically, and to demonstrate that the low iteration count needed for the solution of the preconditioned system is nearly constant, for different scenarios.

In the first experiment, we verify the spectral properties of the preconditioned system on a 2D framework consisting of a stationary cylinder placed at the center of a lid-driven cavity.
In Table \ref{tab:PrecResults},
we measure the spectrum's support $[\lambda_{min},\lambda_{max}]$ and show that $\text{spec}(L^{-1}S_p)\subseteq[1,2]$, as predicted theoretically.
Since the scalar approximation of $C$ is crucial for the efficiency of our method, we also measure the spectrum's support for $L^{-1}\tilde{S}_p$, with $\tilde{S}_p$ from Eq. \eqref{eq:tildeSp}.
In this case, most of the eigenvalues satisfy $\lambda\in[1,2]$, and although some eigenvalues slightly exceed 2, the maximal eigenvalue tends to 2 as the grid resolution increases.
It can be explained by that fact that the quality of the scalar approximation $C\approx \frac{1}{2}I$ improves with the grid resolution, as shown in \cite{goncharuk2023immersed}.

\begin{table}
\centering
\begin{tabular}{c|cccc}
\hline
  \toprule
  \mc{5}{c}{Spectrum's support for the preconditioned system}\\
  \midrule
    Grid size (cells)  & $24\times24$ &  $48\times48$ & $96\times96$ & $192\times192$ \\
  \midrule
$\text{spec}(L^{-1}S_p)\subseteq$ &  [1,1.9997] & [1,1.9997] & [1,1.9996]  & [1,1.9996] \\
$\text{spec}(L^{-1}\tilde{S}_p)\subseteq$ &  [1,1.9944] & [1,1.9986] & [1,2.0099]  & [1,2.0067] \\
  \bottomrule
 \end{tabular}
\caption{The spectrum's support of the preconditioned system $L^{-1}S_p$ and the corresponding system with an approximate Schur complement $L^{-1}\tilde{S}_p$, obtained for a stationary circular cylinder placed at the center of a 2D lid-driven cavity.}
\label{tab:PrecResults}
\end{table}

In the second experiment
we compare the GMRES iteration count required for the solution of the original system
$\tilde{S}_p p' = \bfb$
and the preconditioned system
$L^{-1}\tilde{S}_p p' = L^{-1}\bfb$.
Note that during a full solution of the problem, the right hand side $\bfb$ changes at any time step.
Here, we solved the systems for a random right hand side to demonstrate the improvement of the conditioning in terms of iteration count.
The results shown in Fig. \ref{fig:iterations2D} demonstrate that the preconditioned system can be solved in just a few iterations regardless of the grid size, while the original system requires hundreds of iterations for very small grids and thousands for larger grids.
This scalability property enables performing the simulations on dense 3D grids, as shown in the next experiment.

\begin{figure}
\centering
\includegraphics[width=0.47\textwidth]{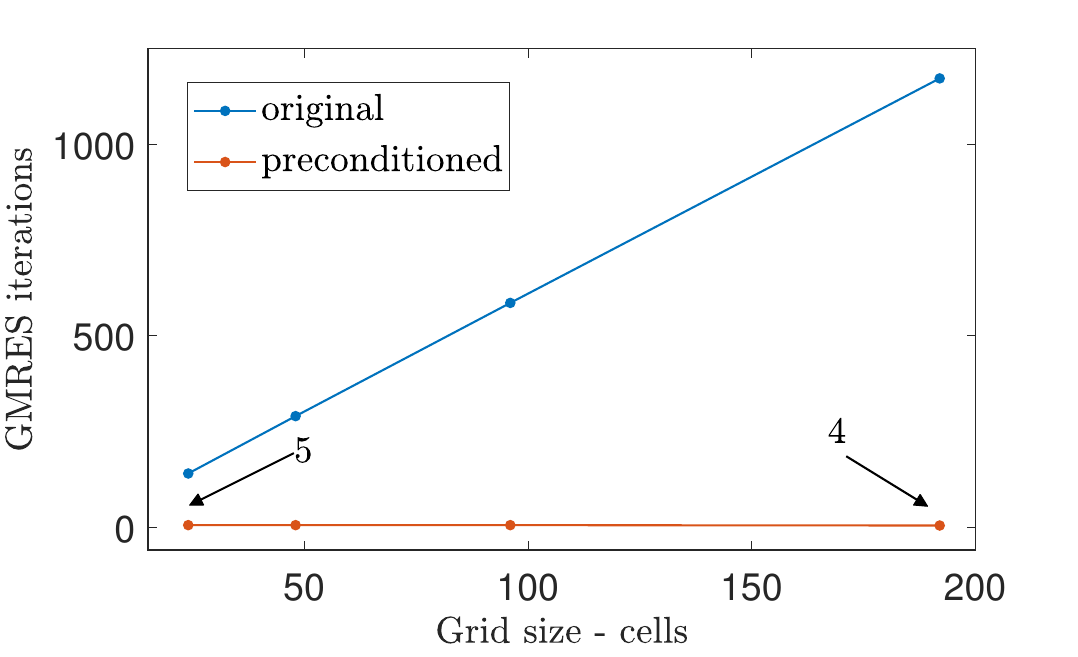}
\caption{GMRES iteration count for the solution of a system with a random right-hand-side for the corresponding matrices.}
\label{fig:iterations2D}
\end{figure}

In the third experiment we verify that the scalability shown above extends to 3D experiments with moving bodies, by demonstrating it on model problem 1 from Subsection \ref{subsubsec:OscillatingSphere}: an oscillating sphere.
We evaluate the efficiency of the proposed preconditioner by measuring the average number of BiCGStab iterations required to converge the preconditioned system to a tolerance of $10^{-12}$ for various Reynolds numbers and grid sizes. Table~\ref{tab:3dResults} summarizes these results, averaged over three amplitude-to-diameter ratios ($A/D = 0.5,\ 1.0,\ 1.5$) and four representative locations of the sphere along its periodic trajectory: bottom dead center, top dead center, and the midpoints during upward and downward motion. The results demonstrate that the average iteration count remains nearly constant—typically between 4 and 5—across all tested Reynolds numbers, grid sizes, and sphere locations. This robustness highlights the effectiveness and scalability of the preconditioner.

\begin{table}[H]
\centering
\begin{tabular}{c|cccccc}
\hline
  \toprule
  \mc{7}{c}{Iteration count for a 3D oscillating sphere}\\
  \midrule
Grid size & $Re=10$ & $Re=20$  & $Re=50$ & $Re=100$ & $Re=200$ & $Re=300$ \\
  \midrule
$50 \times 50 \times 75$ & 4.0 & 4.0 & 4.0 & 4.0 &  4.0 & 4.0 \\
$100 \times 100 \times 150$ & 4.25 & 4.0 & 4.0 & 4.17 &  4.17 & 4.0 \\
$200 \times 200 \times 300$ & 4.33 & 4.25 & 4.17 & 4.33 & 4.33  & 4.67 \\
$400 \times 400 \times 600$ & 4.17 & 4.0 & 4.08 & 4.08 & 4.33 & 4.25 \\
  \bottomrule
 \end{tabular}
\caption{Average BiCGStab iteration count required for convergence of the preconditioned system to a $10^{-12}$ tolerance. Results are averaged over different $A/D$ ratios and over four representative locations of the sphere along its oscillation cycle.}
\label{tab:3dResults}
\end{table}

To further assess the computational performance, we measured the wall-clock time per time step and memory consumption as functions of grid size. These experiments were conducted on a standard Linux server equipped with 64 GB DDR3 shared memory and two Intel Xeon 12C processors (48 threads total). The results, shown in Fig.~\ref{fig:time_memory}, were averaged over the same set of $A/D$ ratios and six Reynolds numbers ($Re = 10,\ 20,\ 50,\ 100,\ 200,\ 300$).
The measured wall-clock time scales approximately as $O(N^{1.1})$, where $N$ is the number of unknowns. This scaling aligns well with the expected complexity of the direct solver used in our implementation estimated as $O(N^{1.33})$ in \cite{lynch1964direct, vitoshkin2013direct}, which is further  improved through the use of the Thomas algorithm in one coordinate direction. Memory consumption is sub-linear and remains modest: even for the largest grid ($400 \times 400 \times 600$, hosting approximately 384 million unknowns), the total memory requirement was only 25 GB. This is a significant improvement over previous work \cite{sela2021semi}, where memory consumption exceeded 128 GB for comparable problem sizes at low Reynolds numbers. Importantly, the computational efficiency and memory consumption are nearly independent of the Reynolds number and time step, provided that the CFL number remains below 0.1.

\begin{figure}[H]
\begin{center}
	\newcommand{\image}[1]{\includegraphics[width=0.47\linewidth]{./#1}}
    \subfigure[Time scalability]{\image{time.eps}\label{fig:time}}
    \hspace{20pt}
    \subfigure[Memory scalability]{\image{memory.eps}\label{fig:memory}}\\
\end{center}
\caption{Wall-clock time per time step and memory consumption as functions of grid size. Results are averaged over three $A/D$ ratios ($0.5,\ 1.0,\ 1.5$) and six Reynolds numbers ($10$ to $300$). The solid line represents a linear regression trend.}
\label{fig:time_memory}
\end{figure}

In the last experiment, we demonstrate the same scalability property on model problem 2 from Subsection \ref{subsubsec:PorousSphere}, which contains multiple spheres.
In the view of Remark \ref{remark:convergenceRate},
the number of immersed bodies might affect the sparsity of $B$, and thus the scattering of the eigenvalues of the preconditioned system.
The aim of this experiment is to demonstrate that the same scalable behaviour holds, regardless of the number of bodies.
Table \ref{tab:3dPorous} shows the BiCGStab iteration count needed to solve Eq. \eqref{eq:SchurPrimalP} with preconditioning, averaged over a single oscillation period, for different Reynolds numbers (Re=50, 100, 200, and 300) and grid resolutions.
The results show that this period-averaged iteration count remains nearly constant, requiring 4-5 iterations for both 7 and 14 sub-sphere configurations.
Furthermore, these period-averaged iteration counts match the results in Table \ref{tab:3dResults} for a single oscillating sphere, demonstrating consistent performance across different configurations.
These numerical results, combined with Theorem \ref{thm:thm}, demonstrate our method's applicability and scalability across a wide range of applications.

\begin{table}[H]
\centering
\begin{tabular}{c|cccc|cccc}
\hline
  \toprule
  \mc{9}{c}{Iteration count for a 3D porous sphere}\\
  \midrule
   & \mc{4}{c|}{7 sub-spheres configuration} & \mc{4}{c}{14 sub-spheres configuration} \\
\small Grid size & \small $Re=50$ & \small $Re=100$  & \small $Re=200$ & \small $Re=300$ & \small $Re=50$ & \small $Re=100$ & \small $Re=200$ & \small $Re=300$ \\
  \midrule
\small $100 \times 100 \times 150$ & 4.59 & 4.48 & 4.66 & 4.64 & 4.74 & 4.74 & 4.79 & 4.76 \\
\small $200 \times 200 \times 300$ & 4.94 & 4.96 & 4.95 & 4.95 & 4.98 & 4.98 & 4.98 & 4.93 \\
\small $400 \times 400 \times 600$ & 4.81 & 4.82 & 4.82 & 4.87 & 4.97 & 4.98 & 4.98 & 5.00 \\
  \bottomrule
 \end{tabular}
\caption{Period-averaged BiCGStab iteration count for solving the preconditioned system at different Reynolds numbers and grid resolutions, with convergence tolerance of $10^{-12}$. Results shown for configurations with 7 and 14 sub-spheres, for $A/D=1$.}
\label{tab:3dPorous}
\end{table}

\section{Conclusion and future work}\label{sec:conclusion}

We present a novel, preconditioned  direct-forcing IBM for efficient moving boundary and two-way coupled FSI simulations, building upon the SIMPLE-IBM framework~\cite{goncharuk2023immersed}. 
To address the coupling between pressure and force-density corrections, we employ a block reduction technique that solves for the pressure correction first, using the Laplacian operator as a preconditioner for the resulting primal Schur complement. 
We rigorously prove and numerically verify that the Laplacian is spectrally equivalent to the primal Schur complement. 
This key observation enables an efficient and robust solution strategy: 
since the pressure Laplacian is independent of the immersed body configuration, it can be factorized once using a fast direct solver~\cite{lynch1964direct} and reused throughout the simulation.

We demonstrate both theoretically and numerically that the performance of our solver is independent of grid resolution and the number of moving bodies, ensuring excellent scalability in terms of computational time and memory consumption. 
As a result, the method maintains low memory requirements, enabling high-fidelity moving boundary simulations that would typically require high-performance computing to run efficiently on standard workstations or even laptops. 
This improves the accessibility of advanced immersed boundary simulations for the broader CFD community. 
In this work, we demonstrated the applicability of the method to both moving boundary and two-way coupled FSI problems. 
The methodology was successfully verified by simulating flows induced by a transversely oscillating sphere, as well as sedimenting and buoyant spherical particles, showing good agreement with experimental and numerical reference data across a wide range of operating parameters. 
The approach was then extended to more complex configurations involving multiple moving bodies, specifically porous spheres modeled as arrays of rigid sub-spheres. 
The simulations revealed distinct flow features associated with different array porosities, including variations in drag coefficients, phase shifts, and vortex evolution,  while preserving numerical stability and accuracy.

Future research should focus on extending the method to more complex two-way coupled FSI problems, particularly for simulating undulatory locomotion and deformable mesoscale porous media. 
The efficiency of our preconditioner does not depend on the specific choice of discrete delta function, which opens opportunities to improve accuracy by using wider-support kernels~\cite{stein2017immersed}. 
However, since our method relies on a scalar approximation of the regularization block, a key challenge will be to identify discrete delta functions with sufficiently wide support for which this scalar approximation remains valid.

Finally, the current implementation employs a direct Poisson solver~\cite{lynch1964direct} designed for homogeneous boundary conditions. 
While this is not restrictive in our context, where the pressure correction equation naturally adopts such conditions, future extensions may benefit from replacing the direct solver with a geometric multigrid V-cycle~\cite{brandt1977multi}. 
Such a strategy, if properly optimized, may achieve linear-time scaling and further enhance the versatility and accuracy of our approach for a wider class of FSI problems.

\section*{CRediT authorship contribution statement}

Rachel Yovel: analysis and related examples, validation, writing.

Eran Treister: supervision, validation, writing.

Yuri Feldman: supervision, software, validation, writing.

\section*{Declaration of competing interest}
The authors declare that they have no known competing financial interests or personal relationships that could have appeared to influence the work reported in this paper.

\section*{Funding}
This research was supported by the Israel Science Foundation, GrantsNo. 656/23, 2746/25, and by the Lynn and William Frankel Center for Computer Science at BGU. The second author is also supported by the Ariane de Rothschild scholarship and by the Kreitman High-tech scholarship. The third author is supported by the Israel Ministry of Energy and Infrastructure (grant No. 222-11-049).

\appendix

\section{Grid-spacing effects on the $R^T R$ matrix--vector product approximation}
\label{A1}
We investigate the effect of the relative spacing between the Eulerian and Lagrangian grids on the validity of the $R^T R$ matrix--vector product approximation by a scaled diagonal matrix. This approximation, based on the assumption of a gradually varying distribution of Lagrangian forces for sufficiently smooth geometries, as established in our previous study~\cite{goncharuk2023immersed}, is critical for achieving high computational efficiency of the developed solver. The purpose is to investigate to what extent the validity of this assumption depends on the overlap of the influence domains of the discrete Dirac delta function~\cite{roma1999adaptive} currently utilized in both the regularization and interpolation operators.

\subsection{Numerical setup}

The study was conducted using the same physical and numerical configuration as in Section~\ref{subsubsec:Sedimentation}, namely sedimentation of a spherical particle in a viscous fluid. The Eulerian grid resolution was fixed to $200 \times 200 \times 400$, while the spacing of the Lagrangian markers was varied systematically. The simulations were initialized with the particle at rest and advanced for the first 100 time steps. For each configuration, the following quantities were evaluated:

\begin{itemize}
  \item the ratio $\Delta h_{Eu}/\Delta h_L$ between the Eulerian and Lagrangian grid spacings,
  \item the minimum, maximum, and average row (or column) sums of the explicitly constructed matrix $R^T R$,
  \item convergence or divergence of the time integration,
  \item the value of sedimentation velocity $v_z$,
  \item the number of BiCG iterations required for convergence of the system governing the pressure and force-density corrections (see Eq.~\ref{eq:Saddle}).
\end{itemize}

\subsection{Stability, accuracy, and convergence}

The results of the convergence and stability study are summarized in Table~\ref{tab:RtR_spacing}. For dense Lagrangian discretizations ($\Delta h_{Eu}/\Delta h_L \geq 1.58$), the support of the discrete delta functions overlaps strongly, leading to near-singular behavior of the operator $R^T R$. In this regime, the approximation $(R^T R)^{-1} \approx 0.5 I$ is no longer valid, and the simulations diverge after a small number of time steps. 

For coarser Lagrangian grids ($\Delta h_{Eu}/\Delta h_L \leq 1.33$), the simulations remain stable. In this regime, the diagonal approximation provides an accurate representation of the matrix--vector product involving $(R^T R)^{-1}$. The accuracy of the obtained results was assessed by computing the relative deviation of the sedimentation velocity $v_z$ after 100 time steps with respect to a reference configuration characterized by approximately equal Eulerian and Lagrangian grid resolutions (highlighted in bold in Table~\ref{tab:RtR_spacing}). For all converged cases, the deviation of $v_z$ remained below approximately $2\%$, indicating that the diagonal approximation does not compromise accuracy provided that the Lagrangian spacing is neither excessively fine nor excessively coarse.

\begin{table*}[t]
\centering
\caption{Influence of the Eulerian--Lagrangian grid spacing on convergence and solution accuracy. The values of $[R^T R]$ (min, max, avg) are identical for both Reynolds numbers. For all converged cases, 4--5 BiCGStab iterations were required per time step. When divergence occurred, it typically led to fatal failures in satisfying the continuity equation and enforcing the no-slip kinematic constraint on the surface of the transversely oscillating sphere.}
\label{tab:RtR_spacing}
\renewcommand{\arraystretch}{1.2}
\begin{tabular}{c c c c c c c}
\hline
\multicolumn{3}{c}{\textbf{Re = 11.6}} &
\multicolumn{3}{c}{\textbf{Re = 31.9}} &
\multirow{2}{*}{\textbf{$[R^TR] (\min,\max,\mathrm{avg})$}} \\
\cline{1-6}
$\Delta h_{Eu}/\Delta h_L$ & Converged & $v_z$ &
$\Delta h_{Eu}/\Delta h_L$ & Converged & $v_z$ & \\

\hline
2.00 & $\times$ & $-0.05478$ &
2.00 & $\times$ & $-0.06166$ &
1.984, 2.026, 2.081 \\

1.58 & $\times$ & $-0.05520$ &
1.58 & $\times$ & $-0.06240$ &
1.230, 1.300, 1.264 \\

1.33 & $\checkmark$ & $-0.10978$ &
1.33 & $\checkmark$ & $-0.06215$ &
0.857, 0.936, 0.900 \\

1.143 & $\checkmark$ & $-0.11103$ &
1.143 & $\checkmark$ & $-0.06220$ &
0.634, 0.697, 0.662 \\

$\approx$\textbf{1.00} & $\checkmark$ & $\boldsymbol{-0.11156}$ &
$\approx$\textbf{1.00} & $\checkmark$ & $\boldsymbol{-0.06339}$ &
\textbf{0.486, 0.534, 0.506} \\

0.89 & $\checkmark$ & $-0.11169$ &
0.89 & $\checkmark$ & $-0.06344$ &
0.385, 0.424, 0.400 \\

0.80 & $\checkmark$ & $-0.11171$ &
0.80 & $\checkmark$ & $-0.06345$ &
0.313, 0.342, 0.324 \\
\hline
\end{tabular}
\end{table*}

For the two densest configurations, stability could be recovered by replacing the constant diagonal approximation with a rescaled form $(R^T R)^{-1} \approx (R^TR_{\mathrm{avg}})^{-1} I$. However, this correction was found to be effective only at higher Reynolds numbers, whereas at $Re = 11.6$ the relative error between the reference solution and the results obtained using the rescaled approximation reached approximately $50\%$. An important observation is that the performance of the preconditioner is essentially independent of the Lagrangian discretization. In all tested cases, including those that eventually diverged, the number of BiCG iterations required for convergence of the pressure–force correction system (Eq.~\ref{eq:Saddle}) remained between 4 and 5. These presented results define the accuracy limits for applying the diagonal approximation of the matrix–vector product $(R^T R)^{-1}$, and identify the practical range of applicability of the proposed method.

\section{Influence of boundary proximity on the preconditioner performance}
\label{A2}
To further assess the limits of applicability of the developed preconditioner, we examine its performance in configurations where the immersed body is located in close proximity to the computational boundary. Several points should be clarified in this context. First, in order to ensure physical consistency of the interpolation and regularization operators, the currently utilized discrete Dirac delta function must satisfy a number of fundamental properties, including consistent interpolation of uniform fields and conservation of linear and angular momentum of moving bodies (see, e.g., Roma~\cite{roma1999adaptive}). These properties require that the delta function be applied symmetrically with respect to the computational grid. In particular, the region to which a physical quantity is interpolated must be surrounded by Eulerian grid points extending at least over the full support of the discrete delta function. Conversely, when regularizing a Lagrangian quantity onto the Eulerian grid, the computational domain must extend by at least one delta-function support width beyond the location of the Lagrangian point.

For this reason, the application of direct-forcing immersed boundary methods incorporating symmetric delta functions requires special treatment in situations where an immersed body approaches a physical boundary or another solid object. A commonly adopted approach is the introduction of short-range conservative repulsion forces (see, e.g., \cite{maury1997manybody, glowinski1999dlm, glowinski1999fictitious}). These forces ensure that Lagrangian points do not approach each other more closely than approximately the support width of the corresponding delta function, thereby preserving the validity of the interpolation and regularization operators. An alternative approach, apparently providing a more physically consistent scenario, is to employ discrete delta functions with a one-sided kernel such as those proposed in \cite{vanella2009moving, haji2019mls, bale2021onesided}. In the present study, we restrict ourselves to the use of symmetric delta functions only.

When utilizing symmetric delta functions together with repulsion forces, there is no fundamental distinction between interactions of an immersed body with the boundary of the computational domain and interactions between two independently moving immersed bodies. In both cases, the repelling mechanism enforces a minimal separation distance. For this reason, explicit simulations of multiple independently moving bodies do not introduce additional conceptual challenges beyond those already present in the body–-boundary interaction case. Taking these considerations into account, we have chosen to focus on configurations in which an immersed body remains in close proximity to a boundary for an extended period of time, or, in the limiting case, throughout the entire simulation. To investigate this scenario, we therefore consider an additional test involving a transversely oscillating sphere whose surface is located at a distance of only two grid cells from the horizontal boundary. This distance corresponds to the minimal separation at which repulsion forces are not yet activated and at the same time guarantees the correct application of the discrete operators $G^T R$ and $R^T G$. All remaining governing parameters, including the Reynolds number and the ratio between the oscillation amplitude, $A$, and the particle diameter, $D$, are chosen to be identical to those used in subsection~\ref{subsubsec:OscillatingSphere} of the main text. All the velocity boundary conditions were set as no-slip.

The results of the period-averaged BiCGStab iteration count for different $Re$ and $A/D$ values are summarized in Table~\ref{tab:iter_single}. The table shows that the number of iterations of the preconditioned solver remains nearly constant, within the range of approximately 4 to 5, for the entire set of operating parameters. This behavior is consistent with that observed for configurations simulated far from boundaries in the main text of the paper. These results verify the efficiency of the developed preconditioned solver also for configurations involving body–-boundary interactions.

\begin{table*}[t]
\centering
\caption{Period-averaged BiCGStab iteration count for solving the preconditioned system at different Reynolds numbers and grid resolutions and $A/D$ ratios, with convergence tolerance of $10^{-12}$. Results shown for the configuration of a single sphere transversely oscillating in a proximity of vertical wall.}
\label{tab:iter_single}
\renewcommand{\arraystretch}{1.2}
\begin{tabular}{c c c c c c c c}
\hline
Amplitude ratio & Grid size & Re=10 & Re=20 & Re=50 & Re=100 & Re=200 & Re=300 \\
\hline
$A/D=0.5$ 
  & 50$\times$50$\times50$    & 4.02 & 4.00 & 4.00 & 4.01 & 4.00 & 4.00 \\
  & 100$\times$100$\times$150 & 4.08 & 4.11 & 4.07 & 4.04 & 4.05 & 4.09 \\
  & 200$\times$200$\times$300 & 4.63 & 4.46 & 4.24 & 4.24 & 4.27 & 4.37 \\
  & 400$\times$400$\times$600 & 4.56 & 4.58 & 4.55 & 4.41 & 4.22 & 4.17 \\
\hline
$A/D=1.0$ 
  & 50$\times$50$\times$50    & 4.01 & 4.01 & 4.00 & 4.00 & 4.00 & 4.00 \\
  & 100$\times$100$\times$150 & 4.22 & 4.15 & 4.13 & 4.14 & 4.24 & 4.26 \\
  & 200$\times$200$\times$300 & 4.81 & 4.70 & 4.51 & 4.44 & 4.56 & 4.70 \\
  & 400$\times$400$\times$600 & 4.20 & 4.48 & 4.51 & 4.54 & 4.38 & 4.73 \\
\hline
$A/D=1.5$ 
  & 50$\times$50$\times$50    & 4.01 & 4.01 & 4.00 & 4.00 & 4.00 & 4.00 \\
  & 100$\times$100$\times$150 & 4.18 & 4.17 & 4.11 & 4.17 & 4.29 & 4.23 \\
  & 200$\times$200$\times$300 & 4.61 & 4.61 & 4.61 & 4.62 & 4.61 & 4.61 \\
  & 400$\times$400$\times$600 & 4.29 & 4.65 & 4.63 & 4.64 & 4.69 & 4.78 \\
\hline
\end{tabular}
\end{table*}

\small
\bibliographystyle{elsarticle-num}
%\bibliography{IBbib}
\color{black}
%\bibliography{IBbib.bib}

\end{document}